\documentclass[nofootinbib, superscriptaddress, notitlepage]{revtex4-1}
\usepackage{graphicx,color}
\usepackage[section]{placeins}
\usepackage{textgreek,stackengine}
\usepackage{amsmath,amssymb,bm}
\usepackage[left=2cm, right=2cm, top=3cm]{geometry}
\graphicspath{ {Images/} }
\usepackage{amsthm}
\usepackage{textcomp}
\usepackage{hyperref}
\usepackage{mathtools}
\usepackage{algorithm}
\usepackage[noend]{algpseudocode}
\usepackage{dsfont}
\usepackage{braket}
\usepackage{appendix}
\usepackage{xcolor}
\definecolor{cornflowerblue}{rgb}{0.39, 0.58, 0.93}
\definecolor{pink1}{rgb}{0.87, 0.36, 0.51}
\definecolor{purple1}{rgb}{0.44, 0.16, 0.39}
\definecolor{orange1}{rgb}{1.0, 0.22, 0.0}
\definecolor{azure}{rgb}{0.0, 0.5, 1.0}
\definecolor{arsenic}{rgb}{0.23, 0.27, 0.29}
\usepackage{xpatch}

\newcommand{\edits}[1]{\textbf{#1}}
\theoremstyle{plain}
\newtheorem{thm}{Theorem}[section]
\theoremstyle{definition}
\newtheorem{defn}{Definition}[section]
\theoremstyle{plain}
\newtheorem{assum}{Assumption}[section]
\theoremstyle{plain}
\newtheorem{lem}{Lemma}[section]

\makeatletter
\xpatchcmd{\algorithmic}{\itemsep\z@}{\itemsep=2ex plus2pt}{}{}
\makeatother

\begin{document}

\title{Randomized Benchmarking in the Analogue Setting}

\author{E.~Derbyshire}
\affiliation{School of Informatics, University of Edinburgh, Edinburgh EH8 9AB, UK}

\author{J.~Yago Malo}
\affiliation{Department of Physics and SUPA, University of Strathclyde, Glasgow G4 0NG, Scotland, UK}

\author{A.~J.~Daley}
\affiliation{Department of Physics and SUPA, University of Strathclyde, Glasgow G4 0NG, Scotland, UK}

\author{E.~Kashefi}
\affiliation{School of Informatics, University of Edinburgh, Edinburgh EH8 9AB, UK}
\affiliation{Laboratoire d\textquotesingle Informatique de Paris 6, CNRS, Sorbonne Universit\'{e}, 4
place Jussieu, 75005 Paris, France}

\author{P.~Wallden}
\affiliation{School of Informatics, University of Edinburgh, Edinburgh EH8 9AB, UK}

\date{\today}

\begin{abstract}
Current development in \emph{programmable} analogue quantum simulators (AQS), whose physical implementation can be realised in the near-term compared to those of large-scale digital quantum computers, highlights the need for robust testing techniques in analogue platforms. Methods to properly certify or benchmark AQS should be efficiently scalable, and also provide a way to deal with errors from state preparation and measurement (SPAM). Up to now, attempts to address this combination of requirements have generally relied on model-specific properties. We put forward a new approach, applying a well-known digital noise characterisation technique called randomized benchmarking (RB) to the analogue setting. RB is a scalable experimental technique that provides a measure of the average error-rate of a gate-set on a quantum hardware, incorporating SPAM errors. We present the original form of digital RB, the necessary alterations to translate it to the analogue setting and introduce the analogue randomized benchmarking protocol (ARB). In ARB we measure the average error-rate per time evolution of a family of Hamiltonians and we illustrate this protocol with two case-studies of analogue models; classically simulating the system by incorporating several physically motivated noise scenarios. We find that for the noise models tested, the data fit with the theoretical predictions and we gain values for the average error rate for differing unitary sets. We compare our protocol with other relevant RB methods, where both advantages (physically motivated unitaries) and disadvantages (difficulty in reversing the time-evolution) are discussed.
\end{abstract}
\maketitle

\section{Introduction}
In the quest for real-world applications in the present and near-future, the focus of quantum computing has fallen on noisy intermediate-scale quantum systems. While there is a lot of interest in digital quantum computing/simulation (DQC), significant progress has been made in analogue  quantum simulation (AQS), which offers existing medium-scale systems, and where relevant physical problems are mimicked by a highly-tunable quantum system. AQS are engineered to evolve continuously in time according to a specific Hamiltonian \emph{or} family of Hamiltonians which can be flexibly controlled and microscopically understood from first principles. Instead of the application of quantum logic gates, as in the digital setting, a calculation is performed through the unitary or dissipative evolution under the system Hamiltonian. As a result, AQS currently lacks many of the existing testing or error-correcting methods of DQC. Being able to trust that AQS are simulating the designated quantum system, or running the same class of Hamiltonians in a \emph{reproducable} way becomes essential when the AQS cannot be classically simulated.  
\par 
Most AQS experiments aim to estimate the ground state of a Hamiltonian or to learn about the dynamical properties of the system and the spread of information, i.e. on a given quantum hardware or when evolving under a given Hamiltonian. Typically, AQS is tested against accurate numerical methods in classically solvable regimes (such as one-dimensional dynamics) or cases where the final outcome is known, i.e. classical optimization. More recently, a number of ideas have been developed for tackling classically intractable regimes. These include a self-verifying technique for simulations of the Lattice-Schwinger model \cite{Selfverifying}, a method based on non-demolition measurements of the Hamiltonian \cite{Yang19} and the broader concept of cross-platform verification, i.e. comparing results for the same problem from different quantum hardware or the measurement of compatible correlation functions that can verify specific Hamiltonians even when the quantum dynamics \emph{cannot} be classically simulated \cite{Hangleiter2017,Bermejo2018}. It is also possible to estimate the final state of the AQS through quantum process/state tomography (QPT/QST) \cite{QST, QPT}, direct fidelity estimation (DFE) \cite{DFE}, or in recent years, matrix product state (MPS) tomography  \cite{MPSTomo} and neural-network approaches \cite{Torlai2018}. \par 
\par
QST is a technique used to reconstruct an unknown quantum state when given multiple copies of it to measure, and to estimate the final state of a quantum simulation, whilst QPT is a similar technique used to reconstruct an unknown quantum process (quantum channel) and to test the process of a quantum simulation. Both techniques require resources that grow exponentially with the size of the system. DFE is a technique to determine whether a system arrives at some target state or runs a target gate. This is more efficient than the previous techniques since knowing the target state/operation means that there are fewer resources (measurement settings) required, however it does not directly account for the errors from state preparation and measurement (SPAM) and is still not scalable in the sense that the \emph{actual} number of measurements still scales exponentially. MPS tomography provides an estimate of the final state of the system, and also requires less measurement settings than standard tomography due to the fact that it exploits tensor network techniques (and specifically matrix product states) to approximate the final state; again, it does not include SPAM errors in its analysis and in the suitable non-classical regime, scales exponentially in terms of resources. To summarise, (i) none of the techniques mentioned previously succeed in giving an account of the errors from state preparation and measurement, meaning that the accuracy to which they characterise the noise of the system is always bounded by an unknown contribution from the SPAM errors and ii) these methods are not efficiently scalable, and can only provide characterisation for systems of up to $\sim20$ qubits failing to capture correlations further than a few sites \cite{MPSTomo} or relying on state symmetries \cite{Toth2010}. 
\par
Despite the considerable advances in recent years \cite{Selfverifying} most of the current methods rely on model-specific properties to improve on previous proposals. The aim of our research is not to verify the ground state estimate of complex Hamiltonians, but instead find a general way to capture the performance of an analogue quantum simulator in running a family of Hamiltonians. Randomized benchmarking (RB) is a digital method to find a measure for the holistic performance of a quantum hardware that takes into account SPAM errors and is theoretically efficiently scalable, i.e. in the length of the gate sequences applied ($poly(L)$) and in the size of the system ($poly(N)$). This scalability is dependent on the efficiency with which the gateset tested can be composed/compiled with the native operations of the chosen hardware. 
\par 
The motivation for this work grew from the need to overcome the limitations of previous methodologies and develop scalable non-model-dependent methods, building on ideas from the digital setting and applying them to AQS, i.e. with the intention to reduce the gap in testing techniques that exists between the digital and analogue fields of quantum computing. RB in the analogue setting would offer a way to determine whether your chosen quantum architecture will reliably perform a set of quantum evolutions or unitaries, giving an average error rate for this set of unitaries; particularly useful when considering
\emph{programmable} analogue quantum simulation. RB has the potential to efficiently provide a measure of performance of AQS in regimes that are currently classically intractable without the SPAM error noise floor on the accuracy of this partial noise characterisation. Moreover, it will become apparent that RB is more natural and physically motivated for analogue systems than for their digital counterparts. It is also important to mention that the method presented below particularly in the case of perfect time inversion is similar to the Loschmidt Echo \cite{Gardiner1997,Gorin2006}, a well-known technique relevant in the context of quantum chaos.
\par 
With these motivations in mind, we propose extending randomized benchmarking to the analogue setting, which we call \textit{analogue randomized benchmarking} (ARB). In Section \ref{sec: RB} we present the original form of randomized benchmarking and the technical details of it, in order to keep this work self-contained and to better explain the alterations made for our contribution. Following in Section \ref{sec: analogue}, the modifications necessary to extend randomized benchmarking to analogue quantum simulators; including the current form of the analogue randomized benchmarking (ARB) protocol (Sec.~\ref{sec:ARBProtocol}). In Section \ref{sec: casestudies} we first present two case studies of simulating ARB on concrete analogue models (Sec.~\ref{sec: simplecase}) and then in Sec.~\ref{sec: furtherexplorations} we analyse the robustness of this characterisation in the presence of more complex noise, dissipation and experimental imperfections. Note that we do not perform experiments on physical hardware, but instead emulate these devices by classically modelling physically motivated noise scenarios. We conclude in Sec.~\ref{sec: conclusion} with considering the barriers for physical implementation of this protocol as well as some ideas to overcome them, and introducing the possible future directions that have transpired through this process. 
\section{Randomized Benchmarking}\label{sec: RB}
Randomized benchmarking (RB) \cite{SNEwRUO, Knill08, AAMeier, Onorati19} is a technique for evaluating the performance of a quantum hardware which simplifies the error channel of a quantum process such that one can extract the average error per gate of a particular gate-set on this hardware. This measure is for each gate when it is run as part of a long random computation, a relevant metric for most quantum algorithms that rely on this type of computation. In its standard form, it relies on the gate-set having a particular distribution, and on being able to efficiently invert a string of gates from the gate-set with one single operator. A general randomized benchmarking protocol will consist of variants of the following steps: 1) preparing an initial state $\rho_\psi = \ket{\psi}\bra{\psi}$ on a quantum device, 2) running very many random sequences of gates of varying lengths such that the system should return to this initial state, 3) measuring the probability that the state remained unchanged, and lastly, 4) plotting and fitting the results to a pre-determined decay model that characterises the average error rate of the gates. 
\par
The central theory is based on two observations, the first is that \emph{any} error channel, irrespective of the correlations it initially exhibits, when \emph{twirled} (essentially averaged over random unitaries, see App.~\ref{sec: additionalsec}) ``behaves'' as a much simpler error channel that is easier to characterise and quantify. Twirling is the process that transforms a quantum channel $\Lambda(\rho)$ into twirled channel $\Lambda_t(\rho)$ by conjugating over unitaries $U(\rho) \in U(d)$, in the following way: 
\begin{equation}
    \Lambda_t(\rho) = \int_U d\mu(U) U \circ \Lambda \circ U^{\dagger}(\rho) \enspace ,
\end{equation}
where $d$ is the dimension of the Hilbert space. If the unitaries $U(\rho)$ are distributed according to the Haar measure $d\mu$ \cite{Haarmeasure}, a measure of uniformity, then the twirled channel becomes a \emph{depolarising channel}. A depolarising channel is a simple quantum channel of the following form: $\mathcal{E}_d(\rho) = p_{\mathcal{E}}\rho + (1 - p_{\mathcal{E}})\frac{\mathds{I}}{d}$. Intuitively, with some probability $p_{\mathcal{E}}$ the depolarising channel $\mathcal{E}_d(\rho)$ leaves the state $\rho$ intact, whilst with the remaining probability $(1- p_{\mathcal{E}})$ returns the maximally mixed state $\frac{\mathds{I}}{d}$; where $\mathds{I}$ is the Identity operator. One of the first papers to suggest this kind of method, by Emerson et al \cite{SNEwRUO}, used the equivalence of a Haar-twirled quantum channel to its average fidelity over Haar-random unitaries, to construct a motion reversal protocol which results in a single parameter (depolarising parameter) to describe a noisy quantum channel. This was introduced as a way to efficiently estimate the strength of the noise channel on a device and with this method, errors from state preparation and measurement (SPAM), which are independent of the length of a sequence, can also be extracted from the error characterisation.  However, this early model had the obstacle that implementing sequences of Haar-random unitaries is inefficient. 
\par 
This brings us onto the second key observation, which is that the averaging over this infinite (Haar) set of unitaries can be mimicked by sampling from a small finite set. Finite sets that approximate this average are known as \emph{unitary t-designs}; randomly sampling from a unitary t-design is equivalent to randomly sampling from the Haar random unitaries provided that the \emph{averaged} quantity computed involves polynomials of order, at most, t. Formally, a unitary t-design is a set of unitaries $\{U_k\}$, where $\{k = 1, ..., K\}$ such that:
\begin{equation}
    \frac{1}{K} \sum_{k=1}^{K} P_{t,t}(U_k) = \int_{U(d)} d\mu (U) P_{t,t}(U) \enspace ,
\end{equation}
for every polynomial $P_{t,t}(U)$ of order $t$, where $d\mu(U)$ is the Haar distribution. In other words, for any polynomial of unitaries of degree $t$ or less, calculating the average over the set $\{U_k\}$ is the same as calculating the average over all the unitaries (Haar integral). An $\epsilon$-approximate t-design is a set of unitaries that has the same property, where it holds only up to some error $\epsilon$. For RB, the gate-set must be \emph{at least} an exact or approximate 2-design (for intuition as to why this is the case, please see App.~\ref{sec: AppA}). In summary, RB aims to test and quantify how well a set $\{U_{k}\}$ performs on average
on a given quantum hardware, by utilising the twirling property of the Haar distribution. The twirling property is then used to simplify the effective error channel and make it quantifiable with a single parameter (depolarisation). In other words, if one has a set that is a unitary t-design, then by performing the method of randomized benchmarking one obtains an average error rate of the said (gate)-set on that quantum device.
\par 
We present here a basic randomized benchmarking protocol, for completeness, and for technical details and further explanation of how the method works we refer to App.~\ref{sec: additionalsec}. 
\begin{figure}
\begin{algorithm}[H]
\floatname{algorithm}{Protocol}
\caption{Digital Randomized Benchmarking}
\label{alg:protocol1}
\renewcommand{\thealgorithm}{}
\begin{algorithmic}[1]
\State Sample uniformly from $\{U_{k}\}$ a number of sequence lengths $S_{l}$ and run a sequence $S_{\eta}$ at length $l \in S_{l}$ where: $S_{\eta} = \Lambda_{U_{k_{\eta}+1}} \Lambda_{U_{k_{1}}}, .... \Lambda_{U_{k_{\eta}}}$, and $\Lambda_{U_{k_{\eta}+1}}$ is a single operator deterministically chosen to invert the preceding sequence of unitaries (i.e. $\Lambda_{U_{k_{\eta} + 1}} = [\Lambda_{U_{k_{1}}},..., \Lambda_{U_{k_{\eta}}}]^{\dagger}$). This sequence should return the system to its initial state $\rho_{\psi}$. \label{a1step1}
\State Repeat this sequence $R$ times and record $Tr[E_{\psi} \ S_{\eta}(\rho_{\psi})]$ to see if initial state $\rho_{\psi}$ survived the sequence $S_{\eta}$ and call this the survival probability $P_{\eta}$ for sequence $\eta$. \label{a1step2}
\State Repeat this for varying sequences of the same length $l$ and find the average probability that the initial state survived for this sequence length, $Tr[E_{\psi} \ S_{l}(\rho_{\psi})]$ where $S_{l}$ represents the average over all sequences of length $l$. Call this the average survival probability for length $l$: $P_{l}$.\label{a1step3}
\State Repeat the above steps for sequences of different lengths, and plot average survival probability against sequence length, i.e. $P_{l}$ vs $l$. Fit results to a pre-determined decay curve: $P_{l} = A + B \ f^l$ where $l$ is the sequence length and $f$ is the fidelity decay parameter, with $A$ and $B$ absorbing SPAM errors. \label{a1step4}
\State The average error rate can be characterised by $r$ where $r = (d-1)(1-f)/d$, and $d$ is the dimension of the Hilbert space for a system of qubit size $n$ ($d=2^n$). \label{a1step5}
\end{algorithmic}
\end{algorithm}
\end{figure} 
In this protocol,  $\Lambda_{U_k}$ is the imperfect implementation of $U_k$, $\rho_\psi$ is the initial state, taking into account preparation errors, whilst $E_\psi$ is the positive-operator valued measure (POVM) element taking into account measurement errors; in the ideal case $E_\psi = \rho_\psi = \ket{\psi}\bra{\psi}$. The probability of the initial state surviving the random sequences is called the \emph{survival probability} where for a random quantum circuit $U_C$, $P := \bra{\psi}\Lambda_{U_C}(\rho_\psi)\ket{\psi} = Tr(E_\psi \Lambda_{U_C} (\rho_\psi))$, with $\Lambda_{U_C}$ the imperfect implementation of $U_C$. The \emph{average} survival probability ($P_l$) in the protocol, is equivalent to a product of twirled depolarising channels and due to the left-invariance of the Haar-twirl, can be compared to the average fidelity of the error channel (see App.~\ref{sec: additionalsec}) which results in the decay curve: $P_l = A + B f ^l$. Therefore when the data is fit to the curve, this gives an average error rate $r$, where $r = 1 - F_{ave} \equiv (d - 1) (1 - f)/d$. The fit parameters $A$ and $B$ vary for different versions of the protocol, and in the simplest case are $A = \frac{1}{d}$, $B = \frac{d-1}{d}$. 
\par
In order for the above protocol to really quantify the average error-rate of a gate-set, as given in step \ref{a1step5}, we need to make a number of simplifying assumptions on the type of noise (error-channels) the device has:  the errors should be (i) gate-independent $\Lambda_{U_k}=\Lambda$ and (ii) time-independent, i.e. independent of the time it takes to run the gate and of when it is applied in any part of any sequence, (iii) the error-channel should be trace-preserving and memoryless (iv) the SPAM errors should be length-independent and (v) the error in the inversion step can be viewed as a single step; therefore, it does not scale with the length of the sequence which allows it to be absorbed into the SPAM errors. Finally, given that the gate-set tested is not universal, it is also crucial that (vi) the inversion step (in the ideal case) can also be implemented using gates from the tested gate-set. Realistically, the physical errors on a quantum process are likely to be gate-dependent and time-dependent, and more recent protocols show RB to be robust against certain types of noise \cite{Wallman17, Mageson11, Mageson12, Wallman15, HashagenFranca, DRB, Merkel18}, however the method is most straightforward when the noise assumptions are adhered to. In our investigation, therefore, we first analyse the simplest case (Sec.~\ref{sec: simplecase}) and build up our protocol with different noise models in Sec.~\ref{sec: furtherexplorations}.
\par 
The Clifford group of operators is a well-known and simple to construct 2-design \cite{AHarrow}, and the average performance of these gates is a relevant parameter for several reasons, including error-correcting codes, and the fact that with the addition of one extra single-qubit unitary gate, the set becomes universal. Therefore, many RB protocols test the Clifford group, utilising the fact that a single inversion operator may be found efficiently for the Clifford group \cite{GottesmanKnill} such that the protocol consists of running strings of random quantum gates where the last gate is the inversion operator. Efficiently finding a single inversion operator in the analogue setting is currently not possible. Due to this, we adopt a method similar to Emerson et al's \cite{SNEwRUO} debut of running imperfect unitaries and their inverses for our analogue randomized benchmarking protocol, rather than a single inverse for a string of gates, where ours is closer to the form of the Loschmidt echo sequence \cite{Gardiner1997,Gorin2006}. Our protocol differs in that we do not test Haar-random unitaries nor even an exact 2-design, instead aiming to benchmark an $\epsilon$-approximate 2-design. There are several RB methods that are relevant to our exploration, which we highlight in the next section.

\section{The Analogue Setting}\label{sec: analogue}
In this section, we first give an overview of extending RB to the analogue setting, with technical details from Sec.~\ref{sec:unitaryset} onwards, and our ARB protocol (see Protocol~\ref{alg:protocol2}) in Sec.~\ref{sec:ARBProtocol}. We replace the quantum logic gates tested by digital RB (most commonly, the Clifford group or generators thereof) with a set of time evolution operators/unitaries: $\{U_{k} = e^{-iH_{k}dt}\}$. One of the key conditions of performing RB with these types of unitaries is that they must converge to an approximate unitary 2-design, i.e. form an $\epsilon$-approximate 2-design. There have been interesting results in benchmarking finite groups that are not the Clifford group or a known 2-design \cite{Helsen18, DRB, HashagenFranca}, in particular Direct \cite{DRB} and Approximate \cite{HashagenFranca} randomized benchmarking; in the former, rather than benchmarking the full Clifford group the protocol tests a set of ``native gates'' with the requirement that they \emph{generate} the Clifford group. The latter writes the RB protocol in terms of arbitrary finite groups giving a bound on results from RB when sampling from an approximate Haar-distribution on this group compared to sampling from the full Haar-distribution. Our protocol is similar in the sense that we aim to sample from an approximate Haar distribution, using a finite set of unitaries, however we do not require that this set forms a group, we only require that this set approximates a 2-design. Therefore, we gain some freedom in our approach which allows us to test more quantum hardware that can not fulfill this stronger condition, such as analogue quantum simulators. The result of \cite{HashagenFranca} is also important to our protocol, because the authors demonstrate that equally meaningful results can be gained from RB with a distribution close enough to that of a 2-design (or fully Haar-random).
\par 
We create a set of unitaries from a disordered set of multi-qubit Hamiltonians (details to follow), in a way that is more easily implementable on a physical device, and we sample unitaries from this set to create long sequences such that we assume after some time that we approximate a 2-design. We view our unitaries as generators of an approximate 2-design and our aim is firstly, to demonstrate heuristically that the disorder we create in our Hamiltonians does in fact result in convergence to an approximate 2-design and secondly, that this convergence rate is sufficient enough for approximate twirling (through ARB) to characterise realistic noise in this setting. Furthermore, no RB scheme has been attempted in the analogue setting nor directly adapted to the natural evolution of a quantum device. One of the caveats of digital one and two-qubit gate RB methods is that they require compilation of physical gates, that are produced naturally in the device, into gates of the Clifford group or other digital gates that generate such a group. Compilation of \emph{any} logical gates that are not native to the device limits the size of the system that one can feasibly test, and we bypass this compilation issue by using only the native capabilities of the device. 
\subsection{Unitary Set for ARB}\label{sec:unitaryset}
The unitaries tested are created by disordering a Hamiltonian $H_s$, native to the quantum device of interest, resulting in a set of Hamiltonians $\{H_k\}$ that when time-evolved produce a set of unitaries $\{U_k = e^{-i H_{k} dt}\}$. Our aim is to construct a technique that while general, takes into account the physical limitations and the possibilities offered by existing systems. For this reason, the choice of $dt$ is determined factoring physical constraints of the device; it should not be less than the minimum time required to \emph{physically allow} for the Hamiltonian to be changed per time-step. We define $H_k$ to be:
\begin{equation}
    H_k = H_s + \zeta_k^{(g, l)} \enspace ,
    \label{eqn: Hk}
\end{equation}
where $\zeta_k^{(g, l)}$ is an added disorder term which we define to be one of the following:
\begin{equation}
    \begin{split}
       \zeta_{k}^{g} &= \Delta_{k}\sum_{ij} \sigma_{i}^{u}\otimes\sigma_{j}^{u} \\
       \zeta_{k}^{l} &= \sum_{ij} \Delta_{k}^{ij} \sigma_{i}^{u}\otimes \sigma_{j}^{u} \enspace,
    \end{split}
    \label{eqn: disorder}
\end{equation}
where the indexes $(g, l)$ denote global ($g$) and local ($l$) disorder terms, and $\Delta_k^{(i)}$ is a disorder potential that varies for every $\{k = 1, ..., K\}$ and, in the case of local disorder, according to which sites it is acting on. Here the index $u = x, y, z$ indicates which product of Pauli operators is to be applied on nearest-neighbours. By varying the original Hamiltonian $H_{s}$ with these disorder terms we generate a family of Hamiltonians.
From this family $\{H_{k}\}$, we generate the time evolution operators to obtain our unitary set $\{U_{k}\}$ for a fixed time-step ($dt$).
\par 
Producing an $\epsilon$-approximate 2-design in this setting is an interesting study, and there are relevant results demonstrating convergence to 2-designs for time-dependent Hamiltonians exhibiting Brownian motion \cite{Kliesch16}, for locally disordered Hamiltonians \cite{BHH} and the use of locally random unitaries to estimate R\'enyi entropies \cite{randomquenches}. In the latter study of R\'enyi entropies the authors found that by applying a local disorder potential and a single-site Pauli-operator to all sites of a Hamiltonian sector ($H|_A$) at each time-step, their random unitaries converged to a unitary 2-design on the sectors, defined from the irreproducible representation decomposition of the group generated by their unitaries. Our Hamiltonians, on the other hand, are time-independent and the aim is to converge to a 2-design over the entire Hilbert space of our system. Their results did however influence our choice of disorder potential and we apply disorder potentials also drawn from a normal distribution, with standard deviation $\delta = J$; however, we apply these potentials globally (homogeneous along the chain) and locally, in order to generate two different sets of unitaries for each model, applying them with a standard spin-spin interaction term. For a fixed Hamiltonian $H_s$, the generated time-evolution operator can never approach a 1-design, let alone a 2-design \cite{Roberts16}, because the distribution of the eigenvalues will be too localised. Disordered Hamiltonians have a rich body of literature exploring their role in quantum thermalization, quantum chaos, scrambling, and random unitaries \cite{Srednicki94, Brown07, scrambling, Guhr90, Swingle16, Hayden07}. Intuitively, for a set of Hamiltonians of the form of Eq.~\ref{eqn: Hk} to converge to a unitary 2-design, the disorder term should not conserve any of the symmetries of the Hamiltonian $H_s$. If the disorder term were to commute with parts of $H_s$ and therefore conserve, for example, the total spin, then the disordered Hamiltonians would be exploring only subspaces of the Hilbert space \cite{Marchildon02}; whereas, breaking all the symmetries of a Hamiltonian should produce statistics associated with random unitaries \cite{Roberts16, AbdElHady02, Guhr97, Blumel92, scrambling, BHH}. Initially, we expected that our sets of unitaries would in themselves be $\epsilon$-approximate 2-designs, and with sufficiently large $K$ we predict this to be the case. For all of our simulations, however, we fix $K = 1000$, as a large finite set of unitaries that is experimentally feasible. Subsequent investigation indicated that the combination of the elements of these sets produce an approximate 2-design after some time (with some sets performing better than others, see Sec.~\ref{sec: casestudies}) hinting that our sets are, rather, generators of $\epsilon$-approximate 2-designs. 
\par 
In ARB the main operation that we want to achieve is twirling the error channel into a depolarising channel by averaging this channel over our unitary set (see Sec.~\ref{sec: RB} and App.~\ref{sec: additionalsec}); with the steps in standard RB performing this necessary twirl because the unitaries that are randomly sampled are \emph{uniformly} distributed. As mentioned, there is evidence that approximating a 2-design is adequate for RB \cite{HashagenFranca} however there is no guarantee that the errors will be depolarised with an $\epsilon$-approximate 2-design since the unitaries will not come from an \emph{exact} uniform distribution. The value of $\epsilon$ is crucial, with only a sufficiently good approximate design producing meaningful results from RB. We assume that our sets produce an $\epsilon$-approximate 2-design for long sequences, and we derive bounds on our results in Sec.~\ref{sec: casestudies} based on this assumption.  Here, we define an $\epsilon$-approximate 2-design and present these bounds.
\begin{defn}
\label{def:two-design}
An $\epsilon$-approximate unitary 2-design defined in terms of the diamond-norm \cite{diamondnorm} is a measure on a finite subset $U(D)$, where $D$ is the dimension of the Hilbert space, that satisfies the following property \cite{ExactandApprxDesigns}:
 \begin{equation}
      \|\mathds{E}_{\alpha} (\Lambda) - \mathds{E}_{\mu}(\Lambda)\|_{\diamond} \leq \ \epsilon \enspace ,
      \\
     \label{eqn: epsilondesigndiamond}
 \end{equation}
 where $\mathds{E}_{\alpha}$ is the twirled channel of $\Lambda$ over a set of unitaries $\{U_{\alpha}\}$ spread according to a probability distribution $\alpha$. $\mathds{E}_{\mu}$ is the Haar-twirl of that channel, with $\mu$ the Haar measure.
 \end{defn}
 \par 
A unitary 2-design corresponds to the case when $\mathds{E}_{\alpha} = \mathds{E}_{\mu}$
. There are a few ways to determine how close a particular set of unitaries is to an exact unitary 2-design, that involve comparison with the way that Haar random unitaries behave, such as the frame potential \cite{spherical} and comparing the second moment operators \cite{HarrowAPPX} (see App.~\ref{sec: app_proof}). 
\begin{thm}
\label{thm: survival-prob-bound}
Let $P_l^\mu = A + B f^l$, $P_l^\alpha = P_l^\mu \pm \delta P_l$ be the average survival probabilities measured from randomized benchmarking for a sequence length, $l$, of unitaries distributed according to the Haar measure, $\mu$, and a set of unitaries, $\{ U_\alpha\}$, distributed according to an unknown $\alpha$, respectively. If the set $\{U_\alpha\}$ is an $\epsilon$-approximate 2-design, it holds that:
\begin{equation}\label{eq:delta_Pl_thm}
    |P_{l}^{\alpha} - P_{l}^{\mu}|=|\delta P_l| \leq l \cdot \epsilon \enspace ,
\end{equation}
where $P_l^\tau = Tr[E_\psi \mathds{E}_\tau(\Lambda)^l(\rho_\psi)], \tau \in \{\alpha, \mu \}$ is the survival probability of input state $\rho_\psi$ when the twirled quantum channel $\mathds{E}_\tau$ is applied to it $l$ times.
\end{thm}
\begin{proof}[Proof Sketch] To go from the \emph{measured} survival probabilities to the exact case, one needs to replace the $l$ sums over the $\epsilon$-approximate 2-design with the corresponding integrals over the Haar measure. For each one of these replacements, the maximum difference between the two probabilities increases by $\epsilon$, resulting in the $l \cdot\epsilon$ at the end of Theorem~\ref{thm: survival-prob-bound}. The full proof is given in App.~\ref{sec: app_proof}.
\end{proof}
\begin{assum}
\label{assum: statisticalerror}
We assume that the statistical error in estimating the value of $P_{l}^{\alpha}$ with the RB method, is much smaller \footnote{This assumption can always be fulfilled with a suitable choice of the number of repetitions $R$ of each sequence.} than the error induced by running RB with an $\epsilon$-approximate 2-design rather than an exact 2-design:
\begin{equation}
    \delta P_{l}^{\alpha} \ll \delta P_{l} \enspace .
\end{equation}
\end{assum}
\begin{thm}
\label{thm: final_bound}
Let $A$ and $B$ be known quantities, and let $r'$ be the average error rate determined from RB for the set of unitaries $\{U_\alpha \}$ that are an $\epsilon$-approximate 2-design. Where the average error rate for Haar-distributed unitaries $\{U_\mu\}$ is $r = (d - 1) (1 - f)$, and $d = 2^n$ where $n$ is the dimension of the Hilbert space the unitaries are applied to. Under Assumption~\ref{assum: statisticalerror} where Theorem~\ref{thm: survival-prob-bound} holds, with $l = 1$ and no errors from state preparation or measurement, $r'$ is bounded as follows:
\begin{equation}
\begin{split}
    r - \epsilon &\leq r' \leq r + \epsilon \enspace .
\end{split}
\label{eqn: epsilonbounds}
\end{equation}
\end{thm}
The proof of Thm.~\ref{thm: final_bound} can be found in App.~\ref{sec: app_proof}.
\subsection{Systematically Inverting Unitaries} \label{sec:timereversal}
The standard form of RB (see Protocol~\ref{alg:protocol1}) involves a single deterministically chosen inversion operator $\Lambda_{U_{k+1}}$ that inverts the preceding unitary sequence. The errors in the process before this inversion step get depolarised through the twirling of the error channel, but the single inversion operator will also have errors attached to it. Due to this inversion operator being much shorter in length than the sequence preceding it, the error associated with this operator is a constant additive
error. While this (SPAM \& inversion) error is not (in general) depolarising, as far as computing the survival probabilities, there exists some depolarising error channel that would have the exact same effect\footnote{This happens because we measure in the basis $\{\ket{\psi}\bra{\psi},I-\ket{\psi}\bra{\psi}\}$ and any off-diagonal terms of the deviation do not contribute in the (survival) probabilities.}, and therefore we can model it as such.
\par 
As mentioned, the inversion operator of any string of Clifford operations can be found efficiently. Unfortunately, there is no equivalent result in the analogue setting, and therefore the initial development of the protocol involves systematically inverting each preceding unitary, like the Loschmidt echo sequence \cite{Gardiner1997,Gorin2006}, and similar to \cite{SNEwRUO}. This systematic inversion of the preceding sequence means that the errors, now a combination of forward evolution and inversion errors
do not, in general, depolarise in the RB process; therefore, for the purposes of analysis in Sec.~\ref{sec: simplecase}, we model the systematically inverted unitaries as perfect. We explore the scenario of having the same type of noise in the inversion operators in Sec.~\ref{sec: furtherexplorations}. A necessary step to physically implementing analogue randomized benchmarking will be to remove the time-reversal aspect of our protocol. Firstly, because with systematic inversion the errors will not necessarily be depolarised, and secondly and, most importantly, because performing this type of time-reversal is not trivial on the analogue quantum simulator. While certain terms such as field or on-site terms can be inverted efficiently in the current experiments, the inversion of off-site couplings or tunneling terms still remain challenging \cite{Garttner2017,Li2017}. To highlight this non-triviality, we note that only recently was such inversion realised even in a simplified model \cite{Lesovik19}.
\subsection{Analogue Randomized Benchmarking Protocol} \label{sec:ARBProtocol}
Determine a set of Hamiltonians $\{H_{k}\}$ with $k = \{1, ..., K\}$ such that $\{{U}_{k}\}$ forms a sufficiently good $\epsilon$-approximate unitary 2-design for twirling over them to approximate a depolarising channel, where ${U}_{k} = e^{-iH_{k}dt}$. The time-step $dt$ is kept the same for each unitary operator to mitigate time-dependent errors (except in Sec.~\ref{sec: furtherexplorations}). We again prepare the quantum system in initial pure state $\rho_{\psi}$ where, if ideally prepared, $\rho_{\psi} = \ket{\psi}\bra{\psi}$ and assume that the error channel is a trace-preserving and memoryless CPTP map, with errors being gate and time-independent. 
\par 
The parameters for the standard RB protocol (see Protocol~\ref{alg:protocol1} in Sec.~\ref{sec: RB}) still apply, but we redefine the form of $S_{l}$ and introduce another measure for the sequence length : $S_{T}$ where $S_{T} \equiv S_{l}$ and $T$ represents the total time to run each sequence of length $l$, i.e. $T = dt \cdot l$ and so, we highlight that the average error-rate gained from this protocol is the average error as a function of physical time $T$. Now, our sequence lengths are called $S_{T}$ and we refer to length as time-length. \par 
  \begin{figure}
\begin{algorithm}[H]
\floatname{algorithm}{Protocol}
\caption{Analogue Randomized Benchmarking}
\label{alg:protocol2}
\renewcommand{\thealgorithm}{}
\begin{algorithmic}[1]
\State Sample uniformly from $\{U_{k}\}$ a number of sequence time-lengths $S_{T}$ and run a sequence $S_{\eta}$ of time-length $T \in S_{T}$ where: $S_{\eta} = [{\Lambda}_{U_{k_{1}}},..., {\Lambda}_{U_{k_{\eta}}}, {\Lambda}_{U_{k_{\eta}}}^{\dagger}, ..., {\Lambda}_{U_{k_{1}}}^{\dagger}$] \ on your system such that if each unitary was perfectly implemented your system would be returned to initial state $\rho_{\psi}$. Here we systematically invert each preceding unitary. \label{a2step1}
\State Repeat the sequence $R$ times; record the survival probability $ P_{\eta} = Tr[E_\psi \ S_{\eta}(\rho_\psi)]$ for this sequence. \label{a2step2}
\State Repeat the above for various sequences of the same time-length to get the average survival probability for each sequence time-length: $P_{T} = Tr[E_\psi \ S_{T} (\rho_\psi)]$, where $S_T$ is the average over all sequences of time-length, $T$. \label{a2step3}
\State Repeat the above steps for sequences of different time-lengths and plot the average survival probability against time-length $P_{T}$ vs $T$. \label{a2step4}
\State Fit the results to a predetermined decay curve of the following or similar form: $P_{T} = A + B  f^{T},$ where again $T$ is the sequence time-length and $f$ represents the fidelity decay parameter of the process, with $A$ and $B$ fit parameters that absorb SPAM errors. And again, the average error rate is characterised by $r$ where $r =(d-1)(1-f)/d$ and $d$ is the dimension of the Hilbert space for a system of qubit size $n$ ($d = 2^{n}$). \label{a2step5} 
\end{algorithmic}
\end{algorithm}
\end{figure}
 In the next section we focus on describing how ARB would operate on an analogue
device, where our analysis involves classically simulating the process, including modelling and implementing the error channels.
\section{Case studies}\label{sec: casestudies}
To analyse ARB, we consider a spin system native to many quantum simulators, which is particularly relevant to the case of trapped ion experiments \cite{Kim2011,Lanyon2011,Britton2012}. Specifically, we modelled the XY Hamiltonian:
\begin{equation}
\label{eq:XY_ham}
H_{s} = \sum^{N}_{ij} J_{ij}(\sigma_{i}^{+}\sigma_{j}^{-} + \sigma_{i}^{-}\sigma_{j}^{+}) + B\sum^{N}_{j} \sigma_{j}^{z} \enspace,
\end{equation}
where $J_{ij} \propto \frac{J}{|i - j|^{\alpha}}$ is the interaction term and \edits{$0 < \alpha < \infty$} dictates the strength of the interaction, $\sigma^{+,-,z}$ are the corresponding Pauli operators, $N$ is the length of the spin chain, i.e. the number of qubits in the system, and $B$ is the transverse magnetic field strength.We focus on two regimes; on the larger $\alpha$ values, where we assume \edits{$\alpha\rightarrow\infty$}, i.e. nearest-neighbour interactions only \cite{Porras2004}, and on $\alpha\sim0$ which corresponds to an all-to-all coupling between sites. Current experimental development \cite{Richerme2014,Jurcevic2014,Bernien2017,MPSTomo} has led to regimes where the theoretical prediction for the above XY Hamiltonian becomes classically intractable, it is therefore essential to provide a characterisation of such devices in the absence of simulation capabilities.
\par
In the following we consider the case of a system consisting of $N=6$ spins studying how the coupling strength, the disorder terms and the field strength affect the protocol. We simulate the evolution of the system from an initial product state $\ket{\psi} = \ket{\uparrow \downarrow \uparrow \downarrow\dots}$, usually denoted as the charge-density wave state. We generate our set of unitaries $\{ U_{k} = e^{-iH_{k}dt}\}$ by adding a disorder term chosen from Eq.~\ref{eqn: disorder} with $k = \{1, \dots, 1000\}$. 
In the first case (Sec.~\ref{sec: simplecase}), we run the ARB protocol with perfect inversion operators, 
and model gate and time-independent noise, adhering to the specifications of the standard RB protocol. We look at the impact of the field-term on the protocol and in the following Sec.~\ref{sec: furtherexplorations}, we explore different noise scenarios such as \emph{spontaneous emission}, \emph{weakly time-dependent noise} and the case where we no longer have perfect inversion operators, and instead model the same type of noise in the backward evolution; as well as the impact of the time-step on our simplest case models.

\subsection{Standard ARB with gate and time-independent noise}\label{sec: simplecase}
We modelled the gate and time-independent noise in the form of uniformly distributed fluctuations to both $J$ and $B$ terms that form the static Hamiltonian $H_{s}$ (see Eq.~\ref{eq:XY_ham}). These terms arise from fluctuations or calibration errors in the trapping fields and are native to the experimental device we model. Furthermore, we run the protocol with no errors in state preparation or measurement. Assuming that we have a sufficiently good $\epsilon$-approximate 2-design, it follows that the decay curve that we expect our data to fit is of the form:
\begin{equation}
      \begin{split}
              P_{T} &= A + Bf^{T} \\
                    &= \frac{1}{d} + \frac{d - 1}{d}f^{T} \enspace,
      \end{split}
      \label{eqn: arbcurve}
\end{equation}
where $d = 2^N = 2^6$. In the following results, we fit the ARB decay curve using a non-linear least squares regression fitting tool \cite{matlabtool}. Assuming no SPAM errors and perfect inverses simplifies the comparison of our results to those of an \textit{exact} unitary 2-design (see Thm.~\ref{thm: final_bound}). We use this to estimate the average error rate by fitting $P_{T}$ from the fidelity decay parameter $f$ as $r = (d - 1)(1 - f)/d$.

\subsubsection{XY Hamiltonian with transverse field (Nearest-Neighbours)}\label{sec:nn}
 \begin{figure}[tb]
  \centering
  \includegraphics[width=12cm]{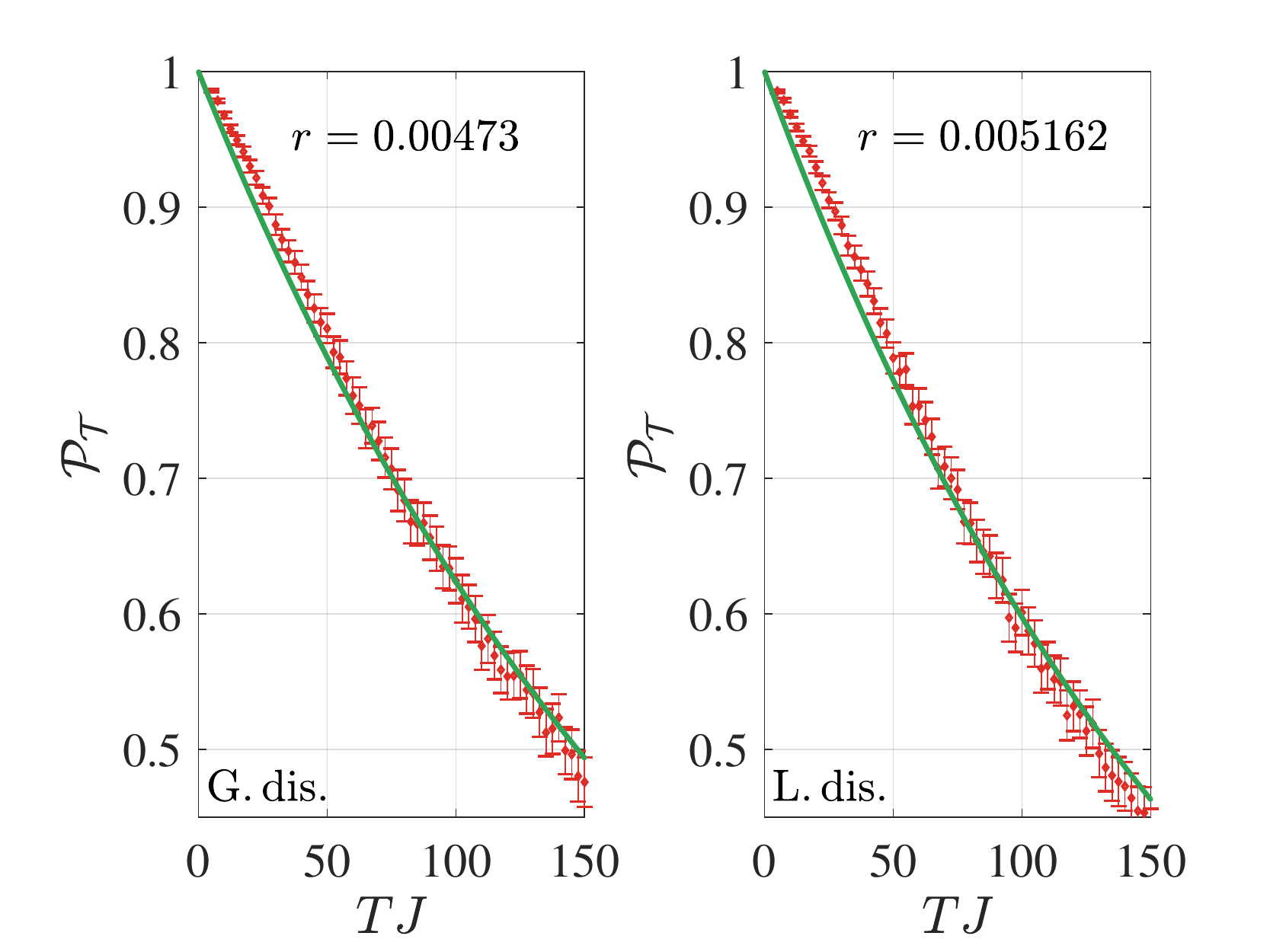} 
  \caption{(Left) Survival probability as a function of total time per sequence, and ARB decay curve fit for a system with nearest-neighbour interactions subject to a global disorder $\zeta_{k}^{g} = \Delta_{k}\sum_{j} \sigma_{j}^{x}\otimes\sigma_{j+1}^{x}$ with  $N=6$, $B=10$, $dt=0.005$, $\Delta B=\Delta J=0$, $\sigma_B=0.5$, $\sigma_J=0.2$, $R=10$ and $n_{seq}=100$; (Right) as with (Left) but subject to a local disorder $\zeta_{k}^{l}= \sum_{j} \Delta_{k}^{j} \sigma_{j}^{x}\otimes \sigma_{j+1}^{x}$. In both cases the resulting fidelity follows the decay curve that we use to estimate the average error rate $r$, which is displayed in each figure. We obtained $r_{g}=0.00473 \ (0.004664, 0.004796)$ 
  for global disorder and $r_{l}=0.005162 \ (0.005068, 0.005256)$ 
  for the local case. For all the displayed results we chose $J=1$ as our frequency reference and scale all timescales with respect to it so that our time axis reads $TJ$.}
  \label{fig:RB_nn}
 \end{figure}
Here, we present the ARB fits for the model described in Eq.~\ref{eq:XY_ham} for the case of $\alpha\rightarrow\infty$ that we consider to be well-described by:
\begin{equation}
\label{eqn: nn-H}
H_{s} = \sum^{N}_{j} J(\sigma_{j}^{+}\sigma_{j+1}^{-} + \sigma_{j}^{-}\sigma_{j+1}^{+}) + B\sum^{N}_{j} \sigma_{j}^{z} \enspace .
\end{equation}
We generate a set of Hamiltonians $\{H_{k}\}$ for both local ($\zeta_{k}^{l}$) and global ($\zeta_{k}^{g}$) disorder terms:
\begin{equation}
\begin{split}
    H_{k}^{g} &= H_{s} + \Delta_{k} \sum_{j} \sigma_{j}^{x} \otimes \sigma_{j + 1}^{x} \\
    H_{k}^{l} &= H_{s} + \sum_{j} \Delta_{k}^{j} \sigma_{j}^{x} \otimes \sigma_{j + 1}^{x} \enspace, \\
\end{split}
\label{eqn: disorderedH-nn}
\end{equation}
where we chose from Eq.~\ref{eqn: disorder} that $u = x$, since this type of coupling should break the symmetry in our $H_{s}$, Eq.~\ref{eqn: nn-H}. This results in a set of $K=1000$ unitaries (with $\Delta^{(j)}_k$ uniformly distributed with standard deviation $\delta=J$) from which to sample for each set $\{H_{k}^{(g,l)}\}$ of the following form:
\begin{equation}
    \{U_{k}^{(g, l)} = e^{-i H_{k}^{(g, l)} dt} \} \enspace . 
\end{equation}

In Fig.~\ref{fig:RB_nn}, we present the ARB protocol results for both the constant global disorder $H_{k}^{g}$ (Left) and the local site-dependent disorder $H_{k}^{l}$ (Right) as indicated in Eq.~\ref{eqn: disorderedH-nn}. On the time evolution forward the system is subject to noise proportional to $H_s$, which we chose to be normally distributed with mean $\Delta J=\Delta B=0$ and standard deviations $\sigma_J=0.2$ and $\sigma_B=0.5$\footnote{Note that these values are taken as an example and we do not require the experimental device to exhibit similar values.}. In these results we conducted $n_{seq} = 100$ sequence iterations for every sequence time-length $S_{T}$ and repeated each individual sequence $R = 10$ times to find the average of a given sequence. We discuss our choices for these parameters in App.~\ref{sec:parameter_conv}. 
\par 
In both cases, the data fits the ARB curve, though at earlier sequence time-lengths the data sits above the curve. The data fitting the curve could imply that the errors are depolarised by the process, which is conditioned on the set being a sufficiently good approximate 2-design. The fact that the data seems to fit to the curve at later times could indicate that the set of unitaries, both global and local, $\{U_k^{(g, l)}\}$ converge to a 2-design at these longer sequences. However, the fit of just one noise model to the curve is not sufficient to imply that our unitaries converge to a 2-design, which is why we explore more complex noise models in Sec.~\ref{sec: furtherexplorations}. In fact, due to our noise model and running perfect inverses, it is possible that the averaging during ARB rather than twirling over an approximate 2-design is what causes the errors to behave like a depolarising channel; hence, we write our results in terms of Thm.~\ref{thm: final_bound}. 
\par
\par
For the average error-rate we obtained (with 95\% confidence bounds):
\begin{eqnarray}
    &r_{l}& = 0.005162\,(0.005068, 0.005256) \\
    &r_{g}& = 0.00473\,(0.004664, 0.004796) \enspace ,
\end{eqnarray}
for locally ($r_{l}$) and globally ($r_{g}$) disordered unitaries, respectively. Where, from Lem.~\ref{lem: boundsonr} we have the average error-rates bounded as follows: \\
\begin{equation}
 \label{eqn:epsilon_rs}   
(r_{\mu} - \epsilon) \leq r_{\alpha}\leq (r_{\mu} + \epsilon) \enspace ,
\end{equation}
where $\alpha\in \{l, g\}$. In this case, $r$ can be thought of as an average infidelity (due to no SPAM errors and no errors in the inversion operators \cite{ProctorWRBM}) per unit time when running a long sequence of these unitaries; therefore, a value of $r \ll 1$ indicates high average fidelity of running the unitaries intended. Since this type of benchmarking has not been applied in the analogue setting before, it is not clear what we would expect the average error rate to be. That is why we set a bound on our results (from Thm.~\ref{thm: final_bound}) and, in Sec.~\ref{sec:alltoall} compare the $r$ value obtained for one of our sets with the error per single unitary for different sets of initial states. It is reasonable to expect a small value of $r$ with the noise model that we chose, and in the context of results gained from more refined noise models (Sec.~\ref{sec: furtherexplorations}) our values of $r_l$ and $r_g$ here seem sensible. 
\par 
As ARB has the potential to be implementable, with some modifications, it is relevant to discuss what $r$ means if we do indeed have a sufficiently good approximate 2-design. We use $r$ to characterise how this type of noisy hardware (modelled as the fluctuations to $J$ and $B$) would behave under this set of unitary operators. In this case, the low values of $r_{(l, g)}$ obtained for our nearest-neighbour model would indicate that both of these sets of operations would perform well overall on such a hardware.

\subsubsection{All-to-all spin model Hamiltonian with transverse field}\label{sec:alltoall}

Here, we present the ARB fits for a spin system governed by Eq.~\ref{eq:XY_ham} for the case of $\alpha\sim0$, i.e. an all-to-all coupled spin system with the same system parameters as in Sec.~\ref{sec:nn}:

\begin{equation}
\label{eqn: aa-H}
    H_{s} = \sum^N_{ij} J_{ij}(\sigma_{i}^{+}\sigma_{j}^{-} + \sigma_{i}^{-}\sigma_{j}^{+}) + B\sum^N_{j} \sigma_{j}^{z} \enspace .
\end{equation}
Again, we generated a set of Hamiltonians $H_{k}$ for both local and global disorder terms:
\begin{equation}
    \begin{split}
         H_{k}^{g} &= H_{s} + \Delta_{k} \sum_{ij} \sigma_{i}^{x} \otimes \sigma_{j}^{x} \\
         H_{k}^{l} &= H_{s} + \sum_{ij} \Delta_{k}^{ij} \sigma_{i}^{x} \otimes \sigma_{j}^{x} \enspace , \\
    \end{split}
\end{equation}
and thereby a set of $1000$ unitaries $\{U_{k}^{(g,l)}\}$ for each $\{H_{k}^{(g,l)}\}$.
\par
 \begin{figure}[tb]
  \centering
  \includegraphics[width=12cm]{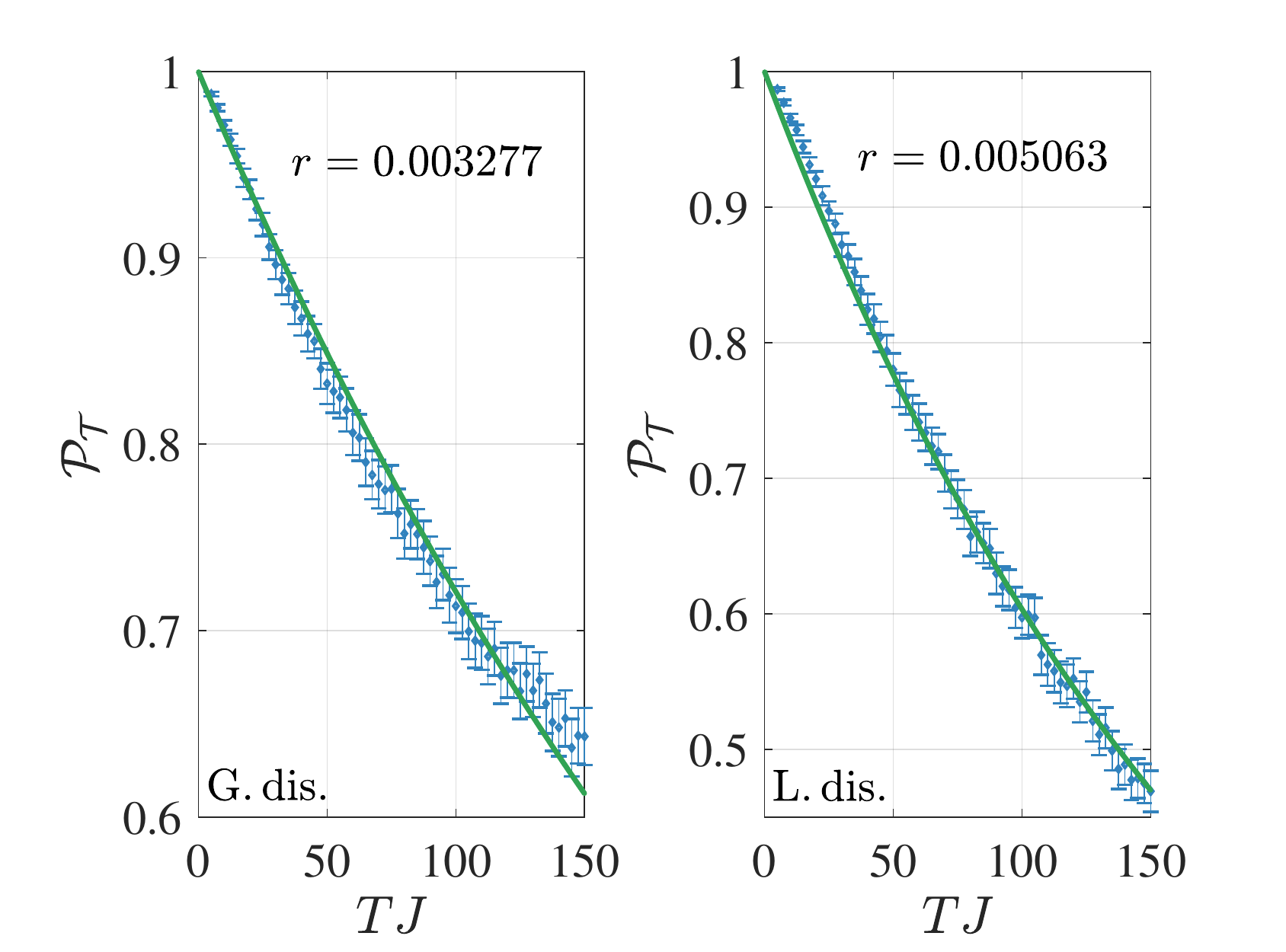} 
  \caption{(Left) Survival probability as a function of total time per sequence, and ARB decay curve fit  for a system with all-to-all coupling subject to a global disorder $\zeta_{k}^{g} = \Delta_{k}\sum_{ij} \sigma_{i}^{x}\otimes\sigma_{j}^{x}$ with  $N=6$, $B=10$, $J=1$, $dt=0.005$, $\Delta B=\Delta J=0$, $\sigma_B=0.5$, $\sigma_J=0.2$, $R=10$ and $n_{seq}=100$; (Right) same as (Left) with a local disorder $\zeta_{k}^{l} = \sum_{ij} \Delta_{k}^{i} \sigma_{i}^{x}\otimes \sigma_{j}^{x}$. In both cases the resulting fidelity follows the decay curve that we use to estimate the average error rate $r$, which is displayed for each case, in the figure. We obtained $r_{g}=0.003277 \ (0.003260, 0.003289)$
  for global disorder and $r_{l}=0.005063 \ (0.005052, 0.005071)$ 
  for the local case.}
  \label{fig:RB_aa}
 \end{figure}
In Fig.~\ref{fig:RB_aa} we present the results of fitting our data for the all-to-all $\{H_k^{(g, l)}\}$ to the ARB curve as in Eq.~\ref{eqn: arbcurve} with the globally disordered unitaries (Left) and locally disordered unitaries (Right). We observe that the two sets display substantially different profiles and, unsurprisingly, values of $r$. In the case of the globally disordered set, the data seems to fit the curve well at very early times then veering further away from the curve until it differs significantly at large $TJ$. This is in contrast to the profile of the locally disordered set, where we observe a much better agreement to the ARB curve. Moreover, the fit of the locally disordered all-to-all set is significantly better than the fit of either of the nearest-neighbour sets (Fig.~\ref{fig:RB_nn}), especially when comparing the behaviour at early times.
\par 

 \par
There are two results to address here, i) comparison of the behaviour of the all-to-all sets $\{U_k^{(g, l)}\}$ and ii) comparison of the behaviour of the locally disordered all-to-all set with both of the nearest-neighbour sets. Before discussing our thoughts on this, we again make it clear that the fit of \emph{one} noise model on these sets does not offer a concrete conclusion as to whether we have an $\epsilon$-approximate 2-design or not. However, if we assume that the data fitting the curve does imply convergence to a 2-design, due to this being a condition of RB producing meaningful results, then we can offer reasons for these results. To address i) let us again consider the scrambling argument discussed in Sec.~\ref{sec:unitaryset}. A global disorder term $\zeta_k^g$ applied to an all-to-all Hamiltonian of the form $H_s$ (Eq.~\ref{eqn: aa-H}) will not necessarily produce enough mixing in the Hamiltonians to explore the Hilbert space sufficiently. 
\par 
To address ii), as discussed earlier we would expect that the more random the disorder we create in our $\{H_k\}$ the faster the convergence to a 2-design with a large set of unitaries $\{U_k\}$, due to \cite{Brown13, scrambling}. The average error rates are as follows:
\par 
\begin{eqnarray}
     &r_{l}&=0.005063 \ (0.005052, 0.005071) \\
    &r_{g}&=0.003277 \ (0.003260, 0.003289) \enspace ,
\end{eqnarray}
with $r_{\alpha}$ bounded by Eq.~\ref{eqn:epsilon_rs}.
Again, if we assume a sufficiently good $\epsilon$-approximate 2-design (an assumption partially substantiated in (Right) Fig.~\ref{fig:RB_aa}) our average error-rates are low. Since we observed the best fit to the curve in the locally disordered all-to-all Hamiltonians, we wanted to determine how sensible this value of $r_l$ is. We numerically determined the average gate infidelity of the set of $K=1000$ unitaries for this noise model, for two sets of $N=100$ random product and pure states, respectively. Since we assume that our set of unitaries converges to a 2-design after some time, it is not expected that the average error rate will be exactly the same; however, for random product states we found a value of $r_l = 0.00150(0.00145,0.0000155)$, which should be representative of the error for the state at the early stages, and for random pure states $r_l = 0.01046(0.01038,0.01054)$. We expect the latter to be higher than the actual value since a random pure state would exhibit substantially more entanglement than a state throughout the time evolution in our protocol. The fact that our value $r_l = 0.005063$ sits between these two values demonstrates that it is a sensible estimate.  
\subsubsection{Impact of Field-term} \label{sec: field-term}
Obtaining the average error-rates ($r$) from the ARB protocol provides a characterisation of the hardware, and how it copes with a specific set of operations (gate-set). We therefore analyse whether some of the \emph{physical} system parameters, of our specific tested system, may impact the measured values of $r$.
\par
\begin{figure}[tb]
	\centering
	\includegraphics[width=11cm]{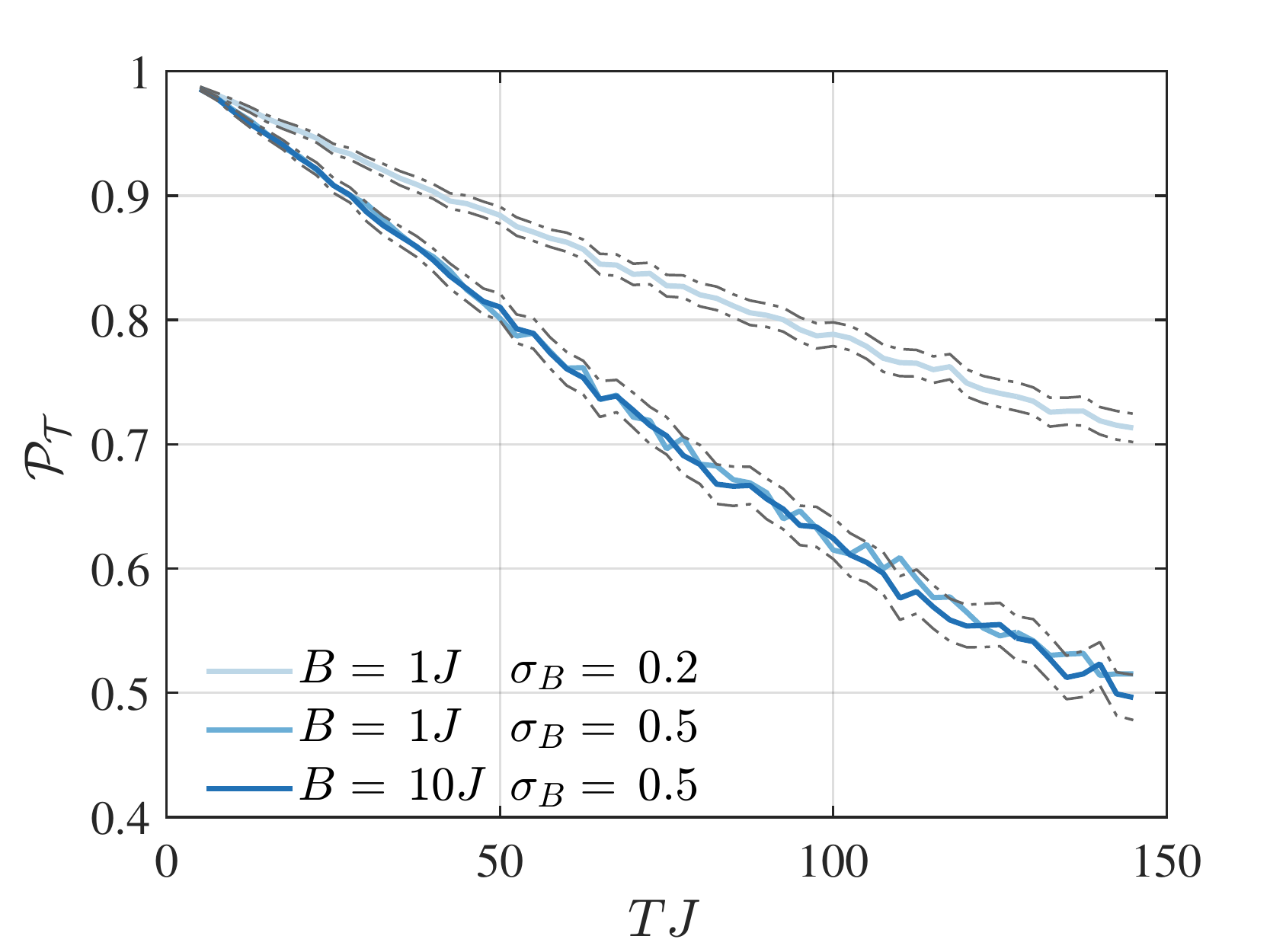} 
	\caption{Survival probability for a system with nearest-neighbour interactions and a global disorder $\zeta_{k}^{g} = \Delta_{k}\sum_{ij} \sigma_{i}^{x}\otimes\sigma_{j}^{x}$ with  $N=6$, $dt=0.005$, $\Delta B=\Delta J=0$, $\sigma_J=0.2$, $R=10$ and $n_{seq}=100$ for different values of field $B$ and the standard deviation of the noise $\sigma_B$. The survival probability does not depend on the static value of the field but only on the magnitude of the noise specified by $\sigma_B$.}
	\label{fig:field_depend}
\end{figure}
In Fig.~\ref{fig:field_depend}, we present the fidelity decay curves for the case of the nearest-neighbour $H_{XY}$ and added global disorder (see Eq.~\ref{eqn: nn-H} and Eq.~\ref{eqn: disorderedH-nn}) as a function of the transverse magnetic field, $B$. We observe that the ARB result does not depend on the off-set value ($B$) of the field but rather only depends on the magnitude of the noise specified by $\sigma_B$. This reveals that, according to our simulations, a quantity that would govern the ground state properties of the device does not affect our protocol. Therefore, the characterisation of the device depends only on the form of the noise and not on the choice of static parameters.
\subsection{Further noise models}\label{sec: furtherexplorations}
In search of more robust claims that our random sets $\{U_k\}$ approximate a 2-design and that ARB can provide meaningful results for more complex noise, we have two main considerations: i) the survival probabilities revealed a better fit to the ARB curve after a finite time (apart from in the globally disordered all-to-all case) which could suggest that the scrambling over time causes the sets to converge to a 2-design. ii) the copious averaging and simple noise model could be the reason for our errors presenting as a depolarising channel, and fitting the ARB curve. We investigate the latter by analysing the protocol in the presence of more elaborate noise: weakly time-dependent noise, spontaneous emission and imperfect inversion operators. These noise models are physically motivated and closer to the experimental conditions of the implementation of the protocol. We also discuss the effects these noise models have on the decay curve in the context of convergence to a 2-design. For these analyses we chose the case of nearest-neighbour coupling with a global disorder term. 
\subsubsection{Weakly time-dependent noise} 
\begin{figure}[tb]
	\centering
	\includegraphics[width=11cm]{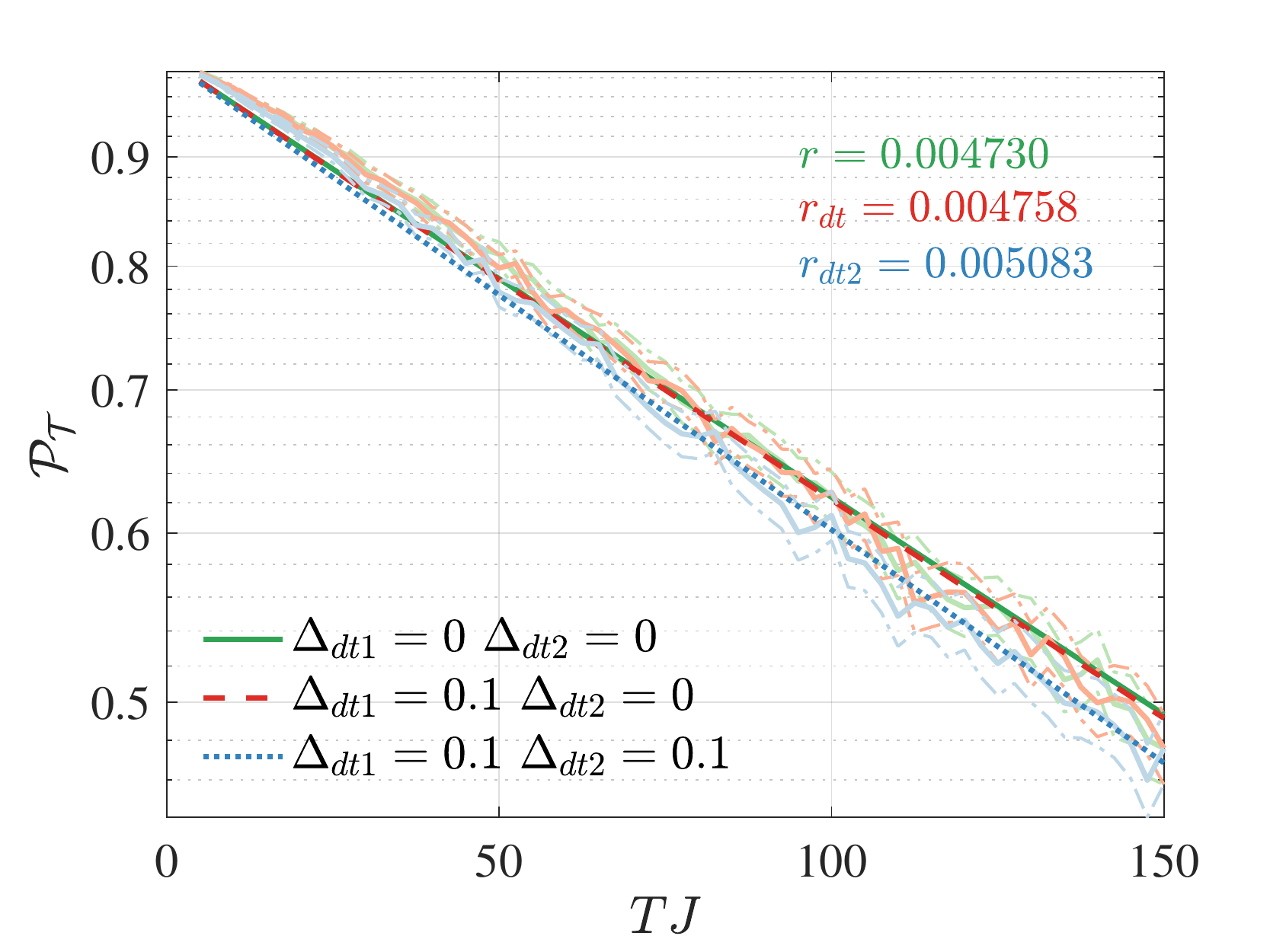} 
	\caption{Survival probability for a system with nearest-neighbour interactions and a global disorder $\zeta_{k}^{g} = \Delta_{k}\sum_{ij} \sigma_{i}^{x}\otimes\sigma_{j}^{x}$ with  $N=6$, $dt=0.005$, $\Delta B=\Delta J=0$, $\sigma_J=0.2$, $R=20$ and $n_{seq}=100$ for an imperfect implementation of the application time of a single unitary with uncertainties forwards, $\Delta_{dt1}$, and backwards, $\Delta_{dt2}$. All results exhibit similar decays and are close to the ARB curve proposed.}
	\label{fig: dtnoise}
\end{figure}
First, in Fig.~\ref{fig: dtnoise} we study how robust our characterization is in the presence of some uncertainty in the time $dt$ for which these random Hamiltonians are evolved and applied on the state. 
This could arise physically from a limited time resolution in the quantum hardware. We consider the case where every unitary is applied for a given time with an uncertainty of $\Delta_{dt1}/dt=0.1$ in the forward evolution whilst inverted with the exact same time-step. We then consider the case when the backwards evolution has the same uncertainty $\Delta_{dt2}/dt=0.1$ and so every $H_k$ is inverted for a slightly different time. We compare these weakly time-dependent models to the fixed $dt$ case, and observe that in all cases the decay profiles remain similar and fit the proposed ARB curve (see Eq.~\ref{eqn: arbcurve}). In the case of $\Delta_{dt2}=0$ (dashed line), since any variation in the forward evolution is matched by the perfect time inversion the results remain unchanged and the estimated average error rate $r_{dt1}$ 
is in the same confidence interval as $r$ from fixed $dt$. Only when this uncertainty between the forward and backwards evolution differs (dotted line) do we observe a drop in the overall survival probability as expected from the unmatched time evolutions in both directions. Despite the presence of this weak time dependence the data exhibits similar agreement with the ARB curve which is a positive indication towards the reliability of our protocol.

\subsubsection{Spontaneous Emission}
As mentioned previously, the affect of the ARB process on one noise model is not enough to indicate that our unitaries depolarise that channel. We therefore consider (again, on our globally disordered nearest-neighbour model) in Fig.~\ref{fig: spontaneousemission} noise from spontaneous emission , an example of coupling of the quantum device to its environment. We model the time evolution of the open quantum system through quantum trajectories \cite{Daley2014}, with a ratio of $\gamma/J = 0.01$, which is compatible with experiments. We expect the average error rate for this noise model to be notably higher, as the dissipation term (which is $\propto\sigma_i^-$) that represents an emission event would transfer the state of the system to states with a much smaller overlap with the initial state. However, since the noise modelled is still chosen to be independent and uniform we expect that the results would fit the ARB curve, Eq.~\ref{eqn: arbcurve}, if all necessary conditions are met by our generated set. We see that the presence of dissipation in the form of spontaneous emission does not impact the profile of the ARB decay curve, and as expected the average error rate is much larger than that found with no dissipation. This is yet another case that can be seen to indicate that the error channel is being simplified and characterised by our protocol.
\begin{figure}[tb]
	\centering
	\includegraphics[width=11cm]{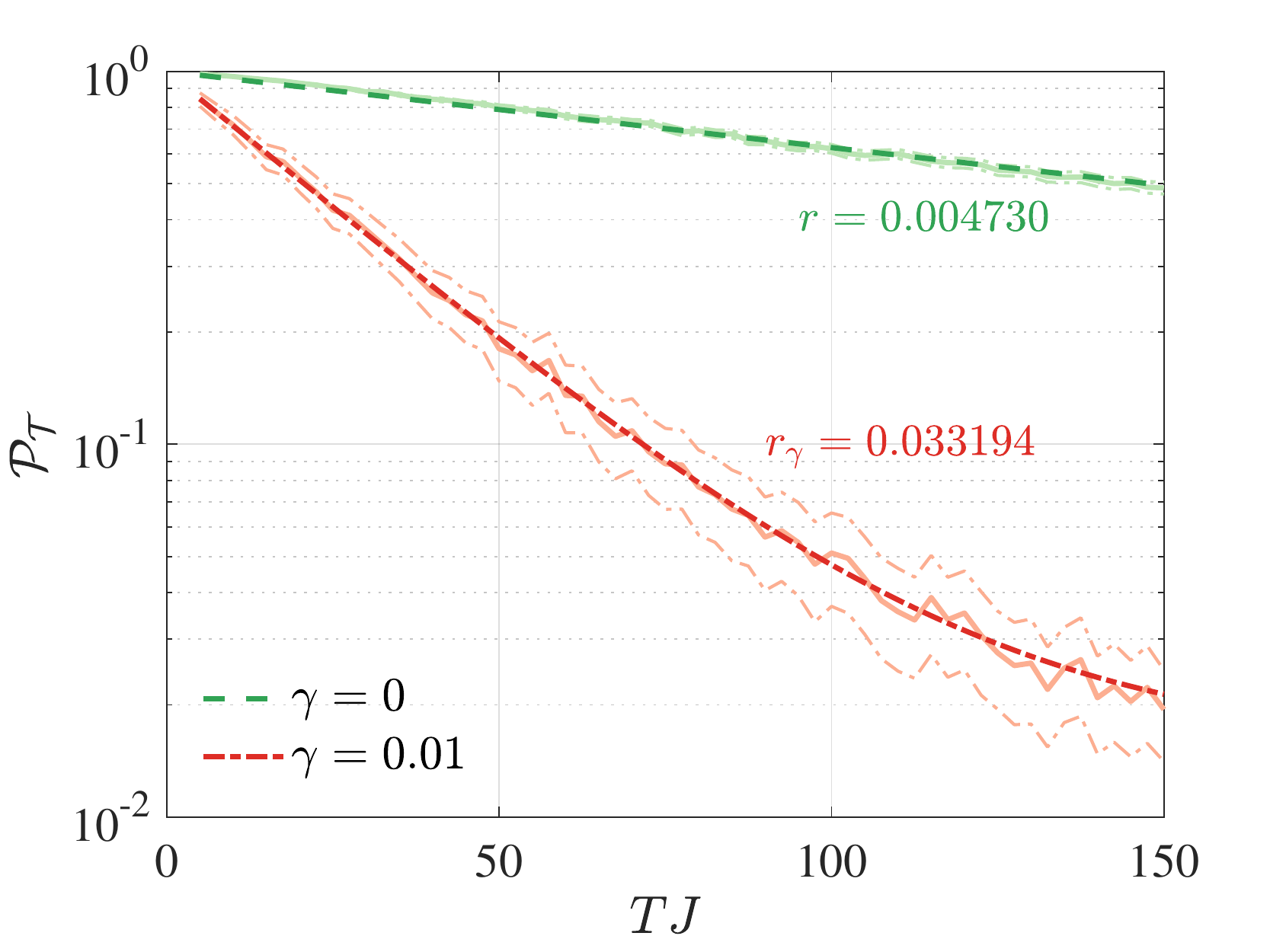} 
	\caption{Survival probability for a system with nearest-neighbour interactions and a global disorder $\zeta_{k}^{g} = \Delta_{k}\sum_{ij} \sigma_{i}^{x}\otimes\sigma_{j}^{x}$ with  $N=6$, $dt=0.005$, $\Delta B=\Delta J=0$, $\sigma_J=0.2$, $R=10$ and $n_{seq}=100$ for a system subject to spontaneous emission ($\propto \sigma_i^-$) with amplitude $\gamma=0.01$. We observe that even in the presence of dissipation the standard ARB curve describes the survival probability profile. The substantial increase in $r$ can be understood from the fact that a single dissipative $\sigma^-$ event would reduce notably the fidelity on a single run of the protocol.}
	\label{fig: spontaneousemission}
\end{figure}

\subsubsection{Noisy time-inversion}\label{sec: timereversalnoise}
Now that we have analysed how the ARB protocol is affected by both weak time-dependent noise and dissipation, it is necessary to address one of the assumptions of the implementation. Namely, in the previous results and those from Sec.~\ref{sec: furtherexplorations} we model the systematic time inversion as perfect. This choice was motivated by the fact that eliminating the time-inversion step is being explored as an extension to the project, in order to make the protocol more experimentally implementable. Nevertheless, this simplification can be regarded as unphysical in the present protocol and we therefore analyse the prospect of a noisy time-inversion. We model the same type of noise (fluctuations to the $J$ and $B$ term) on the inversion operators. In this analysis, we consider both the nearest-neighbour and all-to-all coupling with the globally disordered set only, and compare the results to those in Sec.~\ref{sec: simplecase}.
\begin{figure}[tb]
	\centering
	\includegraphics[width=12cm]{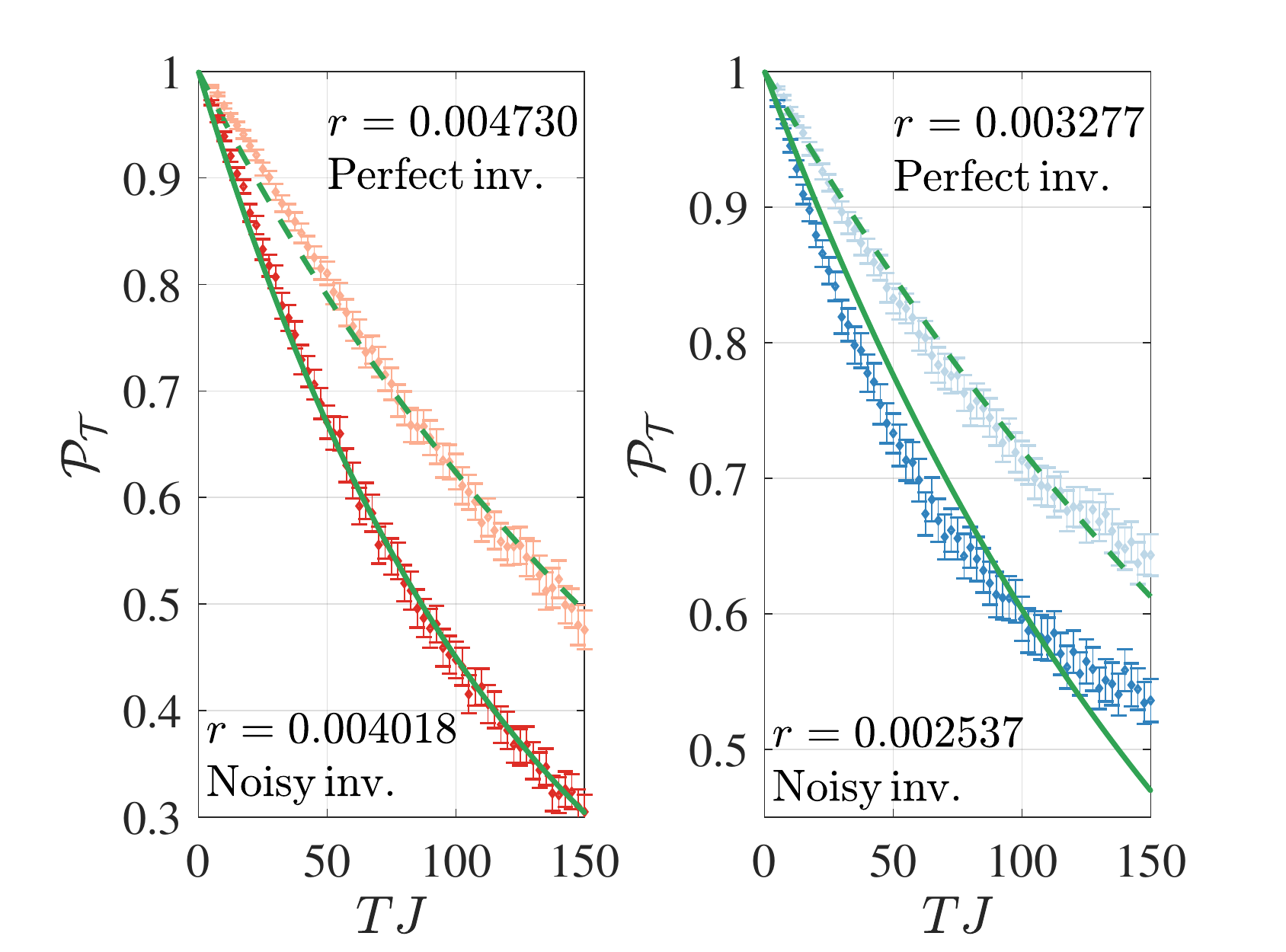} 
	\caption{(Left) Survival probability as a function of total time per sequence, and ARB decay curve fit for a system with nearest-neighbour interactions subject to a global disorder $\zeta_{k}^{g} = \Delta_{k}\sum_{j} \sigma_{j}^{x}\otimes\sigma_{j+1}^{x}$ with  $N=6$, $B=10$, $dt=0.005$, $\Delta B=\Delta J=0$, $\sigma_B=0.5$, $\sigma_J=0.2$, $R=10$ and $n_{seq}=100$. We present both perfect (dashed line) and noisy (full line) time-inversion results, where we have fit the results to the curve $P_{T} = A + B'f^{2(T - 1 \cdot dt)}$  for the noisy inversion, and the original $A + B f^{T}$ for the perfect inversion. (Right) as with (Left) but for an all-to-all coupling model.}
	\label{fig: timereversal_nn_aa}
\end{figure}
In Fig.~\ref{fig: timereversal_nn_aa} we compare the survival probability decay of both ideal and noisy time inversion for nearest-neighbour and all-to-all coupling. 
We fit the noisy inversion results to the following decay curve:
\begin{equation}
    P_T = A + B' f_2^{2(T - 1\cdot dt)} \enspace ,
    \label{eqn: timereversalfit}
\end{equation}
where we set $B' = \frac{d-1}{d}$, following some reasonable assumptions, see App.~\ref{sec: noisyinvappendix}. 
We now have noise in the backward evolution, as well as the forward, therefore adding an additional error channel after each inversion operator. Following the same proof as in the simpler case, with forward noise only, we end up with ($T-1\cdot dt$) twirls of \emph{two} error channels. This means that we have ($T-1\cdot dt$) depolarising channels of the same strength, corresponding to the characterisation of a pair of error channels sequentially applied. We then get additional channels that do not scale with the time of the sequence $T$ (see details in App.~\ref{sec: noisyinvappendix}). In Fig.~\ref{fig: timereversal_nn_aa} and Eq.~\ref{eqn: timereversalfit} we add a factor of two to the curve as we assume that the two composed error channels, when twirled, become a depolarising channel that is equivalent to the product of those two individually depolarised error channels. This is not an unrealistic approximation, and was made by Emerson et al \cite{SNEwRUO} for a more restricted noise model. Interestingly, in the case of nearest-neighbours (Left, Fig.~\ref{fig: timereversal_nn_aa}) the results for noisy time inversion fit this curve much better than the perfect inversion fit its decay curve, which could be due to the factoring out of the SPAM-like errors that creates a closer fit at shorter sequences. We also observe that the average error rate $r_g  = 0.004018(0.003982, 0.004054)$ is roughly 15-16\% smaller than the average error rate from the perfect inversion fit. This is not a huge difference, considering we assume an approximation that was developed for more limited noise, and that we make some approximations to $B'$. In this case, it seems that twirling the combined error channels does not result in the product of those two twirled error channels.  Moreover, strictly speaking, the value we get as $r_g$ can not be directly interpreted as the average gate infidelity since we now have added depolarising channels, and what we actually get here is the average error strength of two error channels (halved).
\par
Comparing the noisy inversion results for the globally disordered all-to-all unitaries (Right, Fig.~\ref{fig: timereversal_nn_aa}), we see that the results veer even further from our derived decay curve. This is not surprising since we theorised that the unitary set for this model did not seem to converge to a 2-design, and twirling over the inverses would not be likely to change this. This noisy decay therefore gives some indication that our theory was correct, since  it differs substantially from the ARB curve, especially at longer times. A final point that is worth mentioning here, is that in the more general and experimentally relevant case, where the inversion error is different than the forward error $\Lambda_{e'}\neq \Lambda_e$, our analysis would again provide an average error rate of a forward and backward evolution ($\Lambda_{e'}\circ \Lambda_e$). This can then be used in different ways. For example, we could assume similar strength of noise in both directions and divide by two, and realistic implementations would have that $\Lambda_{e'} \sim \Lambda_e$ therefore this method would give us a reasonable average error rate. We could also use this to get an upper bound on the noise of the forward evolution, if we were to assume that all the errors came from the forward evolution.

\subsubsection{Impact of time-step}
Finally, we consider how the numerical time-step ($dt$) can impact the ARB curve. In Fig.~\ref{fig:dt_depend}, we present the survival probability $\mathcal{P_T}$ for different values of the time-step ($dt$) used to create the unitaries for the same system as in Sec.~\ref{sec:nn}; again the nearest-neighbour $H_{XY}$ with global disorder. The values of $dt$ were chosen in the regime where the numerical simulations have no impact on the differences between them if the system were noiseless, to avoid any numerical error contribution to the analysis. These results highlight the fact that when a given noise term is applied for a longer period of time in the system it can cause a stronger deviation from the initial state, more difficult to correct with the perfect backwards evolution. In this analysis $dt$ is, therefore, related to the ratio of change of the noise in the actual quantum device, which can affect the result since we model perfect inverses, i.e. noiseless backward evolution. For the purposes of the protocol, that the error on each unitary (gate) should not be dependent on the time it takes to run that unitary, we chose one value of $dt (= 0.005)$ fixed for all time-evolution's simulated. 
\begin{figure}[h]
	\centering
	\includegraphics[width=11cm]{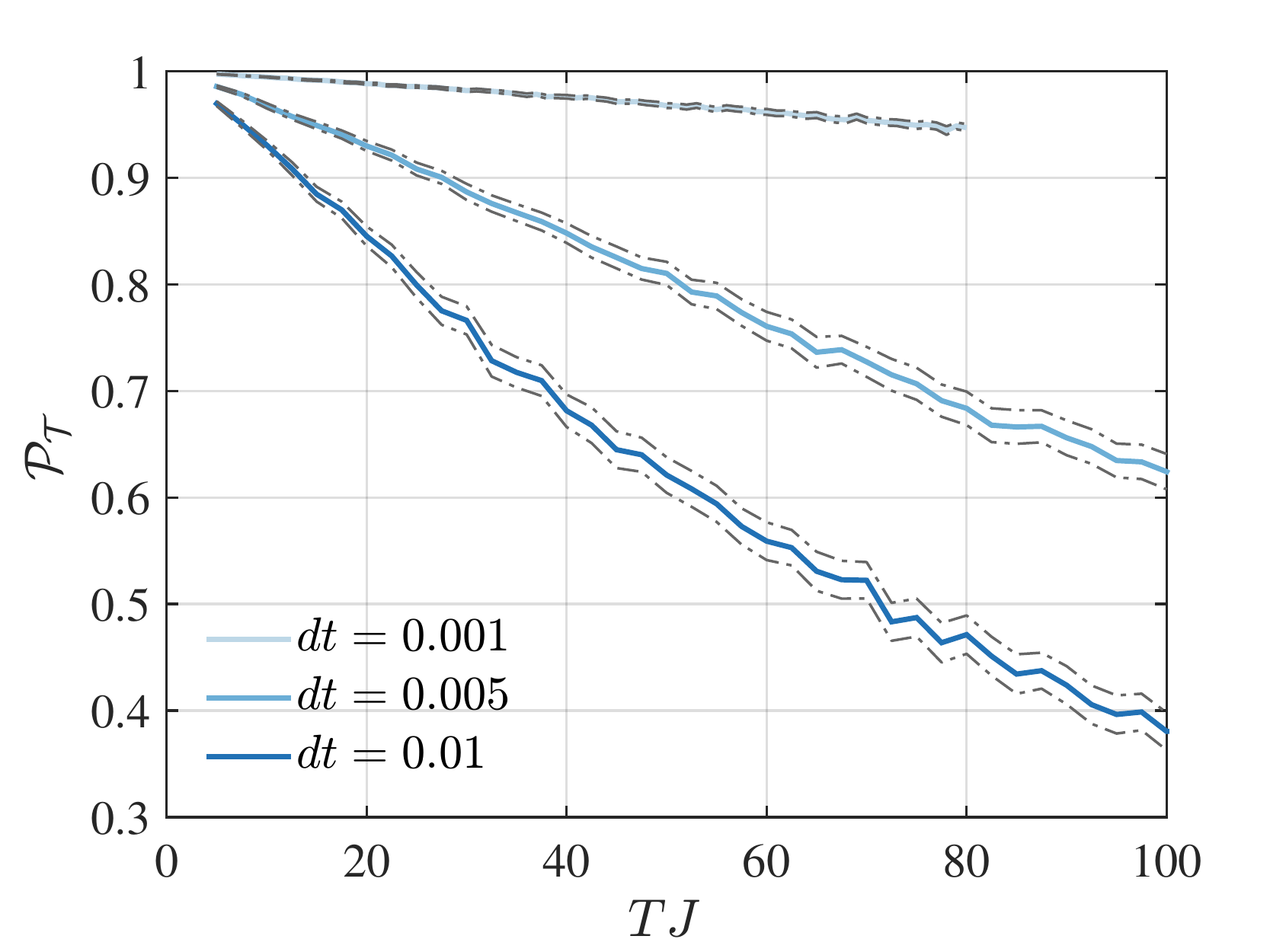} 
	\caption{Survival probability for a system with nearest-neighbour interactions and a global disorder $\zeta_{k}^{g} = \Delta_{k}\sum_{ij} \sigma_{i}^{x}\otimes\sigma_{j}^{x}$ with  $N=6$, $B=10$, $\Delta B=\Delta J=0$, $\sigma_B=0.5$, $\sigma_J=0.2$, $R=10$ and $n_{seq}=100$ for different values of $dt$. We observe that the survival probability decreases as $dt$ increases, since the noise can deviate the system from the initial state for a longer time period, leaving it harder to correct with the noiseless backward evolution.}
	\label{fig:dt_depend}
\end{figure}
\subsection{Results Overview}
The results in Sec.~\ref{sec: furtherexplorations} support the conclusions that were drawn from the simplest noise model, Sec.~\ref{sec: simplecase}. Having observed that the unitary sets in the nearest-neighbour example seemed to converge to an approximate 2-design after some time, $t$, we found that the ARB protocol was robust to more complex noise for the globally disordered nearest-neighbour unitaries which reinforces the notion that this set converges to a 2-design. Furthermore, our simulations in Sec.~\ref{sec: simplecase} reveal that the value of the magnetic field $B$ that governs the state properties of the Hamiltonian, does not affect the results of our protocol; only the form of the noise added to $B$ creates an affect. This supports the idea that our protocol is providing a measure of the noise in our (simulated) system and is robust to changes in the system parameters. The fact that the protocol with these types of unitary sets is robust to weak-time dependent noise, and dissipation is important for experimental implementation. Additionally, we notice that in the case of the globally disordered all-to-all unitaries, which we theorise in Sec.~\ref{sec: simplecase} do not converge to an approximate 2-design, the noisy inversion operators only bring the data further away from the ARB curve. This could indicate that the failure signatures are stronger when more complex noise, that does not necessarily adhere to the standard noise assumptions, is tested. Ultimately, we have created a version of randomized benchmarking with the physical capabilities of analogue quantum simulators in mind, and have found that the adapted protocol displays the behaviour we would expect under the noise tested and supports the notion that some of our sets approximate a 2-design. 

\section{Conclusions and Future work}\label{sec: conclusion}
With the aim of developing a scalable generic method for testing analogue quantum simulators, we extended randomized benchmarking to the analogue setting. By replacing the quantum logic gates in the protocol with unitary time-evolution operators (native to the quantum system) requiring that they converge to a unitary 2-design, fixing the time-step to be the same for each unitary, and systematically inverting the unitaries rather than applying one single inversion operator, we presented the analogue randomized benchmarking protocol. In the context of continuous time evolution, the challenges we met were in: i) creating a set  
of unitaries $\{U_{k}\}$ that generated an $\epsilon$-approximate 2-design and understanding how the convergence rate of these unitaries affects the protocol; ii) the generation of an efficient time-reversal of the unitaries on an analogue system. We numerically simulated our protocol on two models of the XY Hamiltonian (nearest-neighbour and all-to-all), which is native to trapped ions, adding both global and local disorder to generate the random unitary sets. We first modelled uniformly distributed fluctuations (noise) in the coupling $J$ and $B$ field terms of the static Hamiltonian. For the nearest-neighbour sets, the results fit the derived (for this noise) randomized benchmarking fidelity decay curve, particularly in the case of global disorder; this in turn indicated that the sets approximated a 2-design. For the globally disordered all-to-all case, the results did not fit the curve and it seemed that this set was not converging to a 2-design. We found the best fit to the decay curve for the local disorder (all-to-all), which we theorise is due to the richer dynamics of this set and supports the notion that this set converges faster to a 2-design. Therefore, in this case, we compared the average error rate to the average infidelity per gate and found the error rate predicted by our protocol was as expected for this type of gate and time-independent noise. Moreover, the robustness of the ARB decay curve was tested against weakly time-dependent noise, dissipation and an imperfect time-reversal scheme. We observed that for all scenarios of the globally disordered nearest-neighbour set, the proposed decay was suitable to describe the noise channels. For the globally disordered all-to-all unitaries, the imperfect time-reversal revealed further deviation from the curve, giving credence to our interpretation that this set does not converge to a 2-design.
\par 
Analogue randomized benchmarking creates opportunities for improving confidence in analogue quantum simulators by providing alternatives to the current benchmarking techniques. Assuming one has an $\epsilon$-approximate 2-design and the sequences can be efficiently inverted, we could compare the average error-rate $r$ across two quantum devices with the same starting Hamiltonian ($H_{s}$) that the set is built around; since ARB is primarily a test of a specific quantum hardware, $r$ could provide information about what kind of noise were present in each device depending on the results of the protocol on both. Another area that ARB could be useful in is random circuit sampling, where the ARB parameter $r$ could potentially be used to 
prove that random sampling from a random circuit is hard; with future works looking at this direction. At the root, ARB provides a measure for the performance of a set of unitaries on a specific hardware, and in the analogue setting this could be useful in testing programmable analogue quantum simulators. Particularly, the value of $r$ would give a characterisation of how ones device will run a family of Hamiltonians, providing an extra security in the results you would gain from a programmable AQS experiment. 
\par
Extending RB to the analogue setting highlighted many interesting research questions, particularly in regards to approximate unitary t-designs with unitary time-evolution operators. In our work, we assume an $\epsilon$-approximate 2-design is formed from our disordered set of unitaries $\{U_{k}\}$ (formed from disordered Hamiltonians $\{H_{k}\}$) because the disorder added was such that it should be sufficient to break the symmetry of the system Hamiltonian. However, we have not formally proven that our unitary sets $\{U_{k}\}$ are $\epsilon$-approximate 2-designs and we therefore introduced a bound on the results garnered from the ARB protocol. This at least allows us to assess our results for the average error rate within a relative context, and with the standard error on our result we bound the unknown parameter $\epsilon$. Perhaps the RB parameter $r$ could provide an indication of the value of $\epsilon$ for a set of unitaries that categorically are an $\epsilon$-approximate 2-design. An extension to this work is to formally define generating an $\epsilon$-approximate unitary 2-design from a set of unitaries formed around a Hamiltonian native to an AQS. The relations between the frame potential (see Eq.~\ref{eqn: framepotentialunitary}) and the Haar moment operator (see Eq.~\ref{eqn: exactsecondmoment}) \cite{Hun19, Low10} that more accurately characterises an $\epsilon$-approximate 2-design could provide a way
to optimise the generation of approximate designs in the analogue setting. Moreover, exploring the types of disorder that one can add to the starting Hamiltonian, i.e. more locally-addressed, could reveal the optimal type of disorder that generates an $\epsilon$-approximate 2-design with a given Hamiltonian. Furthermore, in this work we have shown that we can provide reasonable fits to the survival probability decay in the analogue setting for small system sizes; however, a relevant point to investigate is how the sampling, both in number of sequences $n_{seq}$ and repetitions per sequence $R$, would exactly scale as a function of the system size.
\par 
Another area to investigate is the limitation of the time-reversal (mentioned in Sec.~\ref{sec:timereversal}) where we have acknowledged that systematic inversion could still provide a measure of the average error and that the main obstacle, in our point of view, to implementing our protocol is the fact that time-reversal in analogue devices is currently not feasible, although it can be implemented for a restricted set of operators, e.g. field terms. For small scale systems, one can compute the ideal output of running sequences on that system and estimate the fidelity of the output state with the ideal state, i.e. using DFE techniques or efficient tomography. This would mitigate the need for the inversion step in our protocol, and the benefit with this type of hybrid technique would be removing the SPAM errors from the characterisation; though, unfortunately, losing the scalability advantage of ARB. On trapped-ion simulators in particular, digital and analogue computations may be performed, and therefore it would be prudent to look at the difference in errors found with both techniques: a possibility for ARB would be to implement the inversion in a Trotterised (digital) way and combine the analogue and digital techniques in order to better characterise the types of errors on this kind of device. 
\par 
The advantages that DRB (Sec. \ref{sec: RB}) and ARB (Sec. \ref{sec: analogue}) have in common are that they evaluate, in a scalable way, the performance of a device whilst also removing the fixed imperfection of the SPAM errors. Comparatively, the use of native gates of the systems means that it is likely that ARB will have smaller errors, e.g. in compilation of gates/more noise that does not adhere to RB assumptions, than digital RB. This could be especially prevalent when dealing with the same physical system used for both analogue and digital quantum simulations. The assumption of gate-independence, and even of \emph{nearly} gate-independence, of the noise model is far better motivated (and closer to reality) in the analogue case which means it is more likely that when experimentally implemented, ARB would give a better fit to the fidelity curve than in the digital case. Moreover, ARB could provide a way to test the performance of digital quantum simulators where researchers could focus on the average error-rate per length of computation time, rather than per-gate. This type of characterisation is not only more physically motivated, but could also bring this analysis closer to the adiabatic model of quantum computation, where complexity is considered in regards to the time taken for the adiabatic evolution.
\section{Acknowledgements}
We thank the anonymous reviewers for their thoughtful comments and time. We thank Ulysse Chabaud, Andreas Elben, Martin Kliesch, Rawad Mezher and Hendrik Waldner for helpful discussions and clarifications. Work at the University of Strathclyde was supported by the EPSRC Programme Grant DesOEQ (EP/P009565/1) and the European Union’s Horizon 2020 research and innovation programme under grant agreement No. 817482 PASQuanS. EK acknowledges support from the following: EPSRC Verification of Quantum Technology grant (EP/N003829/1) and UK Quantum Technology Hub: NQIT grant (EP/M013243/1) and the EU Flagship Quantum Internet Alliance (QIA) project. ED acknowledges support from the Doctoral Training Partnership (EP/N509711/1) under project No.1951737.
\bibliographystyle{apsrev}
\bibliography{Bibliography}
\appendix 
\section{Randomized Benchmarking} \label{sec: additionalsec}
Running a unitary gate $U$ on a physical device corresponds to a quantum channel denoted as $\Lambda_{U}$. The action of this quantum channel can be decomposed in two parts: 
\begin{equation}
    \Lambda_{U} := \Lambda_{U,e} \circ U \enspace , 
\end{equation}
where $\Lambda_{U,e} = \Lambda_{U} \circ U^{\dagger}$, capturing the errors that differentiate $U$ from $\Lambda_{U}$. We use this convention to describe the imperfect channel as one that firstly applies the correct gate ($U$) followed by $\Lambda_{U,e}$, the error superoperator. In our protocol, we define  $\Lambda_{U_{k_{i}}}$ as the imperfect implementation of a chosen unitary $U_{k_{i}}$. 
\par 
To introduce a term commonly used in the literature, the probability that an initial state survives a quantum process is known as the \emph{survival probability}. The survival probability of a channel $\Lambda_{U_C}$, where $U_C$ is \emph{any} unitary circuit, given a fixed initial state $\rho_{\psi}$ is:
\begin{equation}
\label{eqn: traceprobability}
    P := \bra{\psi}\Lambda_{U_C}(\rho_{\psi})\ket{\psi} = Tr\left(E_\psi \Lambda_{U_C}(\rho_\psi)\right) \enspace ,
\end{equation}
where $E_\psi$ is the projection on the state $\ket{\psi}$.
Specifically in RB, we apply a sequence of imperfect unitaries followed by their (imperfect) inverse so that we have:
\begin{equation}
    P=\bra{\psi}\Lambda_{U^\dagger,e}\circ U^\dagger \circ \Lambda_{U,e}\circ U(\rho_\psi)\ket{\psi} \enspace .
\end{equation}
It is clear to see that this probability is equal to unity if both the noise of the forward $\Lambda_{U,e}$ and the backward channels $\Lambda_{U^\dagger,e}$
are the identity (noiseless), since in that case we just evolve the state $\rho_\psi$ by $U^\dagger\circ U=\mathbb{I}$. The average fidelity of the quantum channel (over all pure states) is defined as:
\begin{equation}
   F(\Lambda_U,U)=\int d\psi \bra{\psi}U^\dagger \Lambda_U(\ket{\psi}\bra{\psi})U\ket{\psi} \enspace ,
\end{equation} 
and the average fidelity of a gate-set is given by $\int_U d\mu(U) F(\Lambda_U,U)$. The relevant quantity that we are interested in extracting is the average error-rate of a gate-set (on a specific hardware) which is simply one minus the average fidelity of the gate-set:
\begin{equation}
    r:=1-\int_U d\mu(U) F(\Lambda_U,U) \enspace . 
    \label{eqn: r}
\end{equation}
\noindent\textit{Why it works:} \\ 
Experimentally, we obtain the average survival probability $P_l$ for each length $l$ (step \ref{a1step3}), summing over all the sequences of the same length $l$. By the 2-design property of our unitary set $\{U_{k}\}$ we have:
\begin{eqnarray}
P_l&=&\frac{1}{n_{seq}}\sum_{\textrm{sequences}} \bra{\psi_0}\Lambda_{U_{tot}^\dagger}\Lambda_{U_l}\circ \Lambda_{U_{l-1}}\circ \cdots \circ \Lambda_{U_{1}}(\rho_{\psi_0})\ket{\psi_0}\nonumber\\
&=& \int dU_1\cdots \circ dU_l\bra{\psi_0}\Lambda_{U_{tot}^\dagger}\Lambda_{U_l}\circ \Lambda_{U_{l-1}}\circ\cdots \Lambda_{U_{1}}(\rho_{\psi_0})\ket{\psi_0} \enspace ,
\end{eqnarray}
where $\ket{\psi_0}$ is the initial state of the system.
Note that we have expressed the imperfect inversion operator as a single gate $\Lambda_{U^\dagger_{tot}}$. Decomposing the errors and assuming that they are gate and time-independent $\Lambda_{U,e}=\Lambda_e$, leads to:
\begin{eqnarray} 
\label{eqn: ltimesprobability}
P_l&=& \int dU_1\cdots dU_l\bra{\psi_0}\Lambda_s \circ U_1^\dagger\cdots \circ U_l^\dagger \circ \Lambda_{e}\circ U_l\circ \Lambda_e\circ U_{l-1}\cdots \circ \Lambda_e \circ  U_{1}(\rho_{\psi_0})\ket{\psi_0} \enspace .
\end{eqnarray}
Integrating over $U_l$ twirls one channel $\Lambda_{e}\rightarrow\Lambda_{e,t}$, where $\Lambda_{e,t}$ is the depolarised (twirled) channel corresponding to $\Lambda_{e}$ and the probability that characterises this channel is $p_{e}$ (see Eq.~\ref{eqn: productchannel} and Eq.~\ref{eqn: averageoveru}).
One can then integrate one-by-one the $U_{k}$'s, where each of the integrals result in one error term being twirled and hence depolarised. Noting that the twirled (depolarised) channels commute with all other channels in general and specifically with the unitaries appearing in the above expression, we obtain: 
\begin{equation}
P_l= \bra{\psi_0}\Lambda_{s}\circ (\Lambda_{e,t})^l (\rho_{\psi_0})\ket{\psi_0} \enspace .
\label{eqn: ltimesprob_1}
\end{equation}
Here, $\Lambda_{s}$ represents the error channel corresponding to the SPAM errors. Since the imperfect inverse $\Lambda_{U_{tot}^{\dagger}}$ is one single operator (or at the very least it will be composed of far less gates than the forward sequence) the error associated with it can be absorbed into the SPAM errors. These errors can also be treated as a depolarising channel, because the state is measured in the basis $\{\ket{\psi}\bra{\psi},I-\ket{\psi}\bra{\psi}\}$ and the corresponding ``off-diagonal'' terms do not affect the probabilities that we measure (and need for the subsequent estimations). 
This SPAM error depolarising channel ($\Lambda_{s}$) is characterised by the parameter $p_s$, leading to:
\begin{equation} 
\begin{split}
    P_l &= p_sp_e^l+(1-p_sp_e^l)\frac{1}{d} \\
    &=\frac{1}{d}+(\frac{d-1}{d})p_sp_e^l \enspace ,
    \label{eqn: productchannel}
\end{split}
\end{equation}
which is in the exact form $P_l=A+Bf^l$ mentioned in step \ref{a1step4}, where $f=p_e$, \ $A=1/d$, and \  $B=(\frac{d-1}{d})p_s$. By plotting $P_{l}$ for different values of $l$ we recover the value of $p_e$ ($f$). Having obtained the depolarising probability of the error-channel, we can now look at the average fidelity of the gate-set: 
\begin{equation}
\begin{split}
             \int_{U} d\mu(U) F(\Lambda_{U}, U) &= \int_{U} d\mu(U) F(\Lambda_{U,e}, I) \\
             &= \int_{U} d\mu(U) F(\Lambda_{e}, I) \enspace , \\
\end{split}
\end{equation}
and due to the left-invariance of the Haar measure, we have that $F(\Lambda_{e},I)=F(\Lambda_{e,t},I)$, i.e. the fidelity of any superoperator ($\Lambda_{e}$) with the Identity ($I$) is equal to the fidelity of its exact Haar twirl ($\Lambda_{e,t}$) with the Identity ($I$) \cite{Nielsen02}. Therefore, with the simplifying assumptions made, it is clear to see how the average fidelity is related to $p_e$, since
 \begin{eqnarray}
             \int_{U} d\mu(U) F(\Lambda_{U}, U)                                           
             &=&\int_{U} d\mu(U) F(\Lambda_{e,t}, I)\nonumber\\ 
                                                &=& \int_{U} d\mu(U) ( p_{e} + \frac{1-p_{e}}{d})\nonumber\\
                                                &=& \frac{1}{d} + (\frac{d-1}{d})p_{e} \enspace .
      \label{eqn: averageoveru}
\end{eqnarray} 
Recalling that $r := 1 - F_{ave}$ (see Eq. \ref{eqn: r}) we get the expression of step \ref{a1step5} for the average error-rate of the gate-set: $r= (d-1)(1-p_e)/d$.
\section{Noisy Time-inversion ARB}\label{sec: noisyinvappendix}
Here, we analyse the decay curve in the presence of noisy inversion operators. The sequences that we apply in this scenario are of the following form:
\begin{equation}
    \Lambda_e \circ U_{1}^{-1} \circ \Lambda_e \circ \cdots \Lambda_e \circ U_{l-1}^{-1} \circ \Lambda_e \circ U_{l}^{-1} \circ \Lambda_e \circ U_{l} \circ \Lambda_e \circ U_{l-1} \circ \Lambda_e \circ \cdots \circ \Lambda_e \circ U_{1} \enspace ,
    \label{eqn: noisyinv}
\end{equation}
where for the error channel on both forward and backwards evolution, we assume gate and time-independence, i.e. $\Lambda_{U, e} \equiv \Lambda_e$ and have decomposed the errors, as in Eq.~\ref{eqn: ltimesprobability}. Now, writing the survival probability for this sequence we have:
\begin{equation}
    P_l = \int dU_{1} \cdots dU_{l} \bra{\psi_0} \Lambda_e \circ U_{1}^{-1} \circ \Lambda_e \circ \cdots \Lambda_e \circ U_{l-1}^{-1} \circ \Lambda_e \circ U_{l}^{-1} \circ \Lambda_e \circ U_{l} \circ \Lambda_e \circ U_{l-1} \circ \Lambda_e \circ \cdots \circ \Lambda_e \circ U_{1} (\rho_\psi) \ket{\psi_0} \enspace .
    \label{eqn: survivalprobinv}
\end{equation}
Integrating over $U_l$ results in three depolarising channels. Firstly, the error channel in the ``middle'' of the sequence is twirled over all the unitaries in that space (over all random unitaries in $n_{seq}$ applied here) such that: 
\begin{equation}
    \int_U dU U_{l}^{-1} \circ \Lambda_e \circ U_{l} \rightarrow (\Lambda_e)_t = p_e + (1 - p_e) \frac{1}{d} \enspace ,
\end{equation}
where $(\Lambda_e)_t$ is the depolarised (twirled) channel corresponding to $\Lambda_e$ and the probability that characterises that channel is $p_e$. For the rest of the twirls, it is two error channels that become depolarised, in the following way:
\begin{equation}
    \int_U dU U_{l-1}^{-1} \circ \Lambda_e \circ \Lambda_e \circ U_{l-1} \enspace ,
\end{equation}
which results in $l-1$ depolarising channels of the form $(\Lambda_e \circ \Lambda_e)_t = p_{2} + (1 - p_{2}) \frac{1}{d}$. The last depolarising channel comes from the error channel composed on the first inversion operator (far left in Eq.~\ref{eqn: noisyinv}). These errors can be seen as SPAM errors, and can be treated as a depolarising channel in the same way, since they are independent of sequence length. We can denote the parameter that characterises this depolarising channel as $p_e$. The resulting survival probability for sequences of this form is:
\begin{equation}
    \begin{split}
        P_l &= p_e \ p \ p_{2}^{l-1} + (1 - p_e \ p \ p_{2}^{l-1}) \frac{1}{d} \\ 
        &= \frac{1}{d} + (\frac{d-1}{d}) p_e \ p \ p_{2}^{l-1} \\ \enspace ,
    \end{split}
\end{equation}
which is in the form of Eq.~\ref{eqn: timereversalfit}, where $B' = \frac{d-1}{d} p_e \ p$ and $p_{2} = f_{2}$. For simplicity and easier direct comparison with the perfect inversion, we will assume that $p = p_e = 1$. The actual values are very close to unity. We could have instead left these parameters as variables to be extracted from fitting the curve, something that would have a low impact on our results.
The approximation we make in Sec.~\ref{sec: timereversalnoise} is that the twirl of two composed error channels is the same as the product of those error channels individually twirled (the square of the twirled error channel), i.e. $(\Lambda_e \circ \Lambda_e)_t = (\Lambda_e)_t (\Lambda_e)_t = (\Lambda_e)_t^{2}$. Therefore, in our fit we add the factor of two to the curve which gives a more direct comparison to the parameter $r$ from the original curve.  
\section{Unitary 2-designs}\label{sec: AppA}
Consider a superoperator $\Lambda$ acting on a space $M_{D}$ of D-dimensional quantum states, when $t = 2$, $\Lambda$ has a $D^{2} \times D^{2}$ dimensional matrix representation. We now define a set $U(D)$ of unitary matrices on this space $M_{D}$. If the set is a unitary $2$-design then the space $M_{D}$  is reducible to two irreducible invariant subspaces. Now, we define $\Lambda$ acting on a quantum operator $X$ as: $\Lambda(X) = AXB$. Schur's Lemma \cite{Schur} implies the following (reducible representation) for $U(D)$-invariant trace-preserving operators:
\begin{equation}
    \Lambda(X) = p X + (1 - p) Tr(X) \frac{\mathds{I}}{D} \enspace ,
\end{equation}
where $p = \frac{Tr(\Lambda)-1}{D^{2}-1}$. Considering the fact that a unitary 2-design means that sampling uniformly from the set $\{U_{1}, ..., U_{K}\}$ is operationally equivalent to sampling from the Haar measure, we can say that \cite{ExactandApprxDesigns}: 
\begin{equation}
    \frac{1}{K}\sum_{k=1}^{K} U_{k}^{\dagger} A U_{k} X U_{k}^{\dagger} B U_{k} = \int_{U(D)} dU U^{\dagger} A U X U^{\dagger} B U \enspace ,
\end{equation}
for all $A, X, B \in L(\mathds{C}^{D})$. This essentially means that if we have a set $\{U_{k}\}$ that is a unitary 2-design or above, then conjugating a quantum channel over this set and averaging will result in a depolarisation of that channel.

\section{Parameters Convergence}\label{sec:parameter_conv}
\begin{figure}[!htb]
  \centering
  \includegraphics[width=10cm]{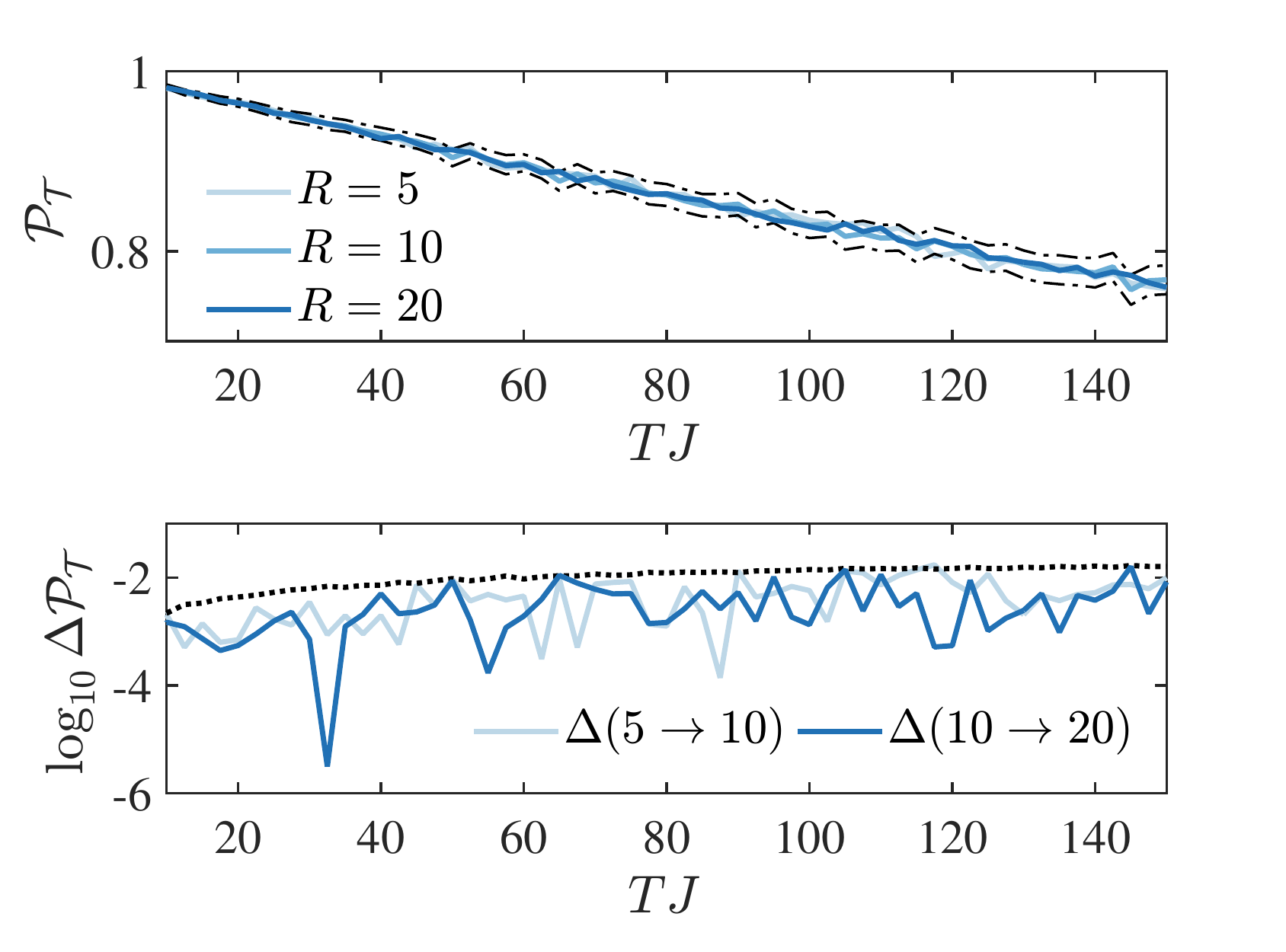} 
  \caption{(Top) Comparison of ARB decay curve for a system with $M=6$, $B=10$, $J=1$, $dt=0.005$, $\Delta B=\Delta J=0$, $\sigma_B=0.5$, $\sigma_J=0.2$ and $n_{seq}=100$ for varying repetitions of the same random sequence $R=5,10,20$. We observe how the average results overlap and remain within the errorbar interval for all R values. We display the error for $R=10$ only to help visualizing the curves; (Bottom) Logarithmic differences between $R=5,10$ and $R=10,20$ curves, we compare this with the error bar of $R=10$ (dotted line). We observe that differences with $R$ are smaller than the statistical uncertainly of the curves.}
  \label{fig:Fig_Rcon}
 \end{figure}
 
Changes to some of the method parameters, such as 
the repeated runs of each random sequence ($R$) and the number of sequences tested for each sequence length ($n_{seq}$) would improve 
the accuracy of our results: iterations over as many sequences as possible of the same length are desired in RB to sample as uniformly from the 
unitary space as possible and repeating each sequence sufficiently gives a more accurate measure of the average survival probability for that sequence. 
In this section we discuss the choice of numerical parameters in the results presented in the main text. Here we describe how the RB curves depend on some of the averaging parameters such as $R$ and $n_{seq}$.
\begin{figure}[!htb]
  \centering 
  \includegraphics[width=10cm]{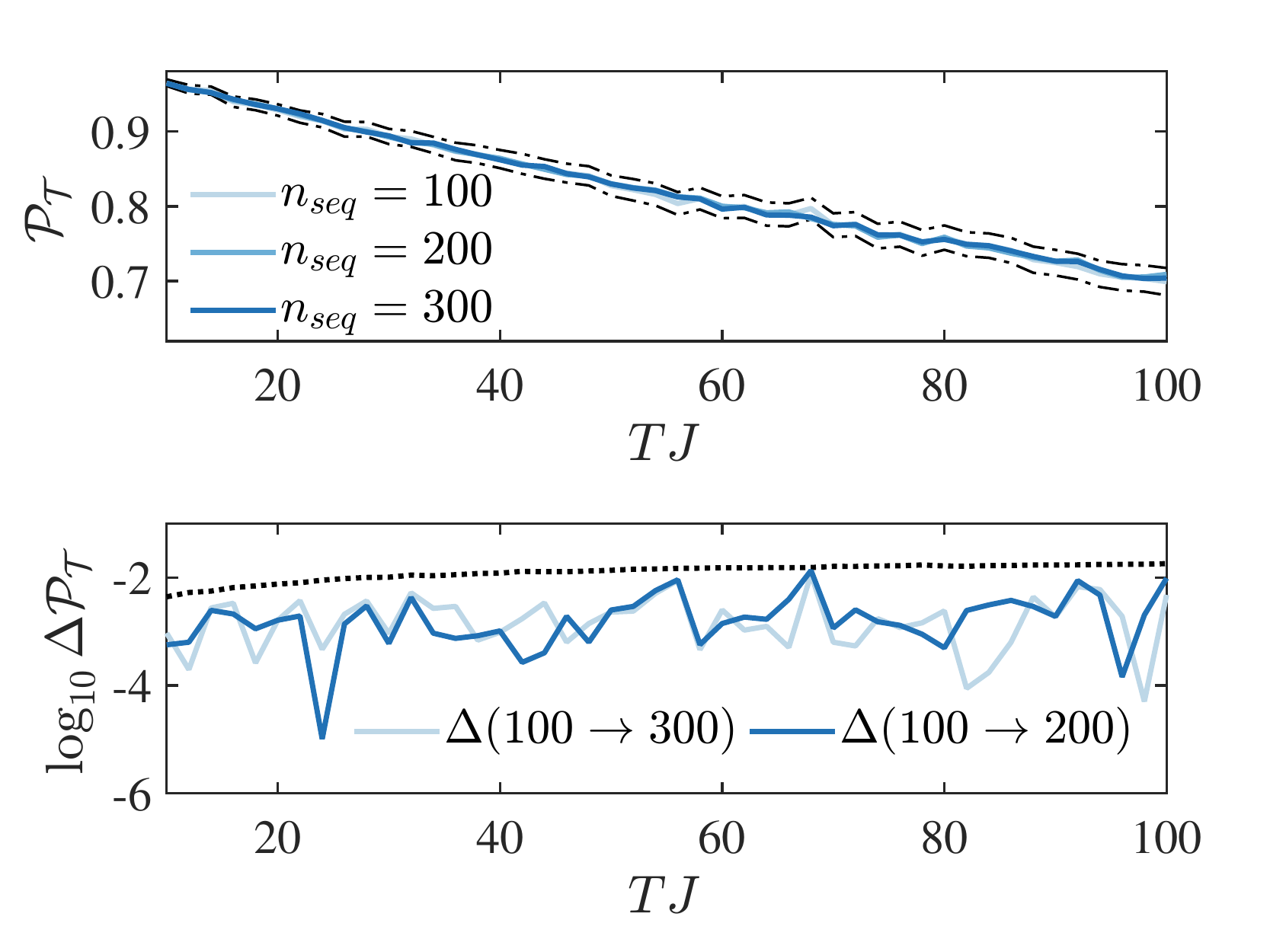} 
  \caption{(Top) Comparison of ARB decay curve for a system with $M=6$, $B=10$, $J=1$, $dt=0.005$, $\Delta B=\Delta J=0$, $\sigma_B=0.5$, $\sigma_J=0.2$ and $R=10$ for varying $n_{seq}=100,200,300$. Similarly, to the case of $R$, we observe good convergence with the range of $S_T$ studied and the average result are within the error intervals. We display the error for $n_{seq}=100$  only to help visualizing the curves; (Bottom) Logarithmic differences between $n_{seq}=100,300$  and $n_{seq}=100,200$  curves, we compare this with the error bar of $n_{seq}=100$  (dotted line). The differences between the curves lay again under the statistical uncertainty.}
  \label{fig:Fig_nseqcon}
 \end{figure}
 \par
 In Fig.~\ref{fig:Fig_Rcon} and Fig.~\ref{fig:Fig_nseqcon}, we justify the choice of the numerical parameters $n_{seq}=100$ and $R=10$ in the main text by comparing the ARB decay curves for varying values of the mentioned variables. We observe that in both cases the statistical uncertainty derived from the sequence averaging is larger than the discrepancy as we vary these parameters, therefore we are confident that the presented results do not depend on the chosen values for $n_{seq}$ and $R$.

\section{Comparative Techniques for 2-designs}
In addition to our discussion in Sec.~\ref{sec:unitaryset}, here we highlight some of the comparative techniques used to determine whether one has an exact unitary 2-design beginning with the introduction of the Spherical t-design.
\par 
Consider a real function $f$, and imagine we are interested in the average value of this function on an $n$-dimensional real sphere $S^{n}$; this is hard to compute, so one can think of averaging over a finite set of unit vectors $D = \{ \ket{\phi_{1}}, \dots , \ket{\phi_{K}}\}$ instead. Briefly, a spherical t-design is a finite subset $D$ of $S^{n}$ such that the average of every $t$-th order polynomial $p$ over $S^{n}$ is equal to the average of $p$ over $D$. For spherical t-designs, the frame potential \cite{spherical} is a well-known metric for determining whether one has an exact spherical design or not, and is defined as follows: \\
\begin{defn} [\emph{Spherical t-design}]
 A set of vectors $\{\ket{\phi_{1}}, \dots , \ket{\phi_{K}}\}$ is a spherical 2-design in $\mathds{C}^{d}$ if and only if:
 \begin{equation}
     \sum_{k, k'} \frac{|\bra{\phi_k}\ket{\phi_{k'}}|^{4}}{K^{2}} = \frac{2}{d^{4} + d^{2}} \enspace .
 \end{equation}
 \end{defn}
The definition of spherical $t$-designs was modified for the unitary setting, changing the real sphere $S^{n}$ to a set of unitaries $U(D)$ and comparing with the Haar distribution, with the term \textit{unitary t-design} coined by Dankert et al \cite{unitarytdesignterm}. 
\par 
Adapting the frame potential from the spherical setting to the unitary setting, D. Gross et al \cite{framepotential} developed the frame potential technique for determining whether ones set of unitaries is an exact unitary 2-design or not, see Definition~\ref{defn: framepotential}. Another technique for determining whether the (suspected) unitary 2-design that one has in an exact unitary 2-design, is to compare the moment operators up to order $t$ (in this case 2) of the particular unitary set and the Haar measure, we refer to \cite{HarrowAPPX} for Definition~\ref{defn: secondorder}.
\begin{defn} [\emph{Frame Potential}] \label{defn: framepotential}
Let the set $\mathds{M} = \{U_{k}\}$ with $\{k = 1, ..., K\}$ be a set of unitaries. The frame potential of $\mathds{M}$ is defined as:
\begin{equation}
    P(\mathds{M}) = \sum_{U_{k}, U_{k'} \in \mathds{M}} \frac{| tr(U_{k}^{\dagger} U_{k'})|^{4}}{K^{2}} \enspace ,
    \label{eqn: framepotentialunitary}
\end{equation}
$\mathds{M}$ is an exact unitary 2-design $\iff P(\mathds{M}) = 2$.
 \end{defn}
\begin{defn}[\emph{Second order moment operators}] \label{defn: secondorder}
    A degree ($t,t$)-monomial in $C \in U((\mathds{C}^{d})^{\otimes n})$ is degree $t$ in the entries of $C$ and degree $t$ in the entries of $C^{*}$. Setting $t = 2$, and collecting all these monomials into a single matrix of dimension $d^{2n2}$ by defining $C^{\otimes 2, 2} := C^{\otimes 2} \otimes C^{* \otimes 2}$, we state that $\alpha$ is an exact unitary $2$-design if expectations of all $(2,2)$ moments of $\alpha$ match those of the Haar measure: 
\begin{equation}
    G_{\alpha}^{(2)} = \mathds{E}_{C \sim \alpha} [C^{\otimes 2} \otimes C^{*\otimes 2}]
    \enspace .
    \label{eqn: secondmoment}
\end{equation}
Therefore, $\mu$ is an exact unitary $2$-design if and only if:
\begin{equation}
    G_{\alpha}^{(2)} = G_{\mu}^{(2)} \enspace ,
    \label{eqn: exactsecondmoment}
\end{equation} 
Where $\mu$ is the Haar distribution.
\end{defn}

\section{Proofs for Section\texorpdfstring{~\ref{sec:unitaryset}}{III. A}}\label{sec: app_proof}
\subsection{Proof of Theorem\texorpdfstring{~\ref{thm: survival-prob-bound}}{III.1}}
In order to prove Theorem.~\ref{thm: survival-prob-bound} we first present the following definitions: 
\begin{defn}
\label{def:trace-norm}
The trace-norm of a quantum channel $\mathds{E}$ in terms of the input state density matrix $\rho$, that minimises the error probability on distinguishing between two quantum channels $\mathds{E}_{1}$ and $\mathds{E}_{2}$, is defined as:
\begin{equation}
    \|\mathds{E}\|_{1} := \max_{\rho} \|\mathds{E}(\rho)\|_{1} \enspace ,
\end{equation}
where $\| \cdot \|_{1}$ is the trace-norm, i.e. $\|X\|_{1} = Tr \sqrt{X^{\dagger} X}$. 
\end{defn}
\begin{defn}
\label{def:diamond-norm}
The diamond-norm distance written in terms of the trace-norm of a quantum channel $\mathds{E}$ is as follows:
\begin{equation}
\begin{split}
    \|\mathds{E}\|_{\diamond} &= \| I \otimes \mathds{E}\|_{1} \\
    &\geq \|\mathds{E}\|_{1} \enspace .
\end{split}
\end{equation}
\end{defn}
\begin{defn}
\label{def:surv-prob-ed}
The average survival probability for each sequence length found from RB is:
\begin{equation}
    P^{\mu}_{l} = A + B f^{l} \enspace ,
\end{equation}
with a pure input state $\rho$, and using an \textit{exact} unitary 2-design, where the average error rate of the unitaries $\{U_{\mu}\}$ is $r = (d-1)(1-f)/d$ (see conversations surrounding, and including, Eq. \ref{eqn: averageoveru}). \par 
Similarly, the average survival probability for each sequence length, measured for an unknown $\alpha$-distribution of unitaries via RB, is:
\begin{equation}
\label{eqn: deltasurv-prob-ed}
    P^{\alpha}_{l} := P_l^\mu \pm \delta P_l \enspace ,
    \end{equation}
where the unitaries are assumed to be an $\epsilon$-approximate 2-design and $P_{l}^{\tau} = Tr[E_{\psi} \mathds{E}_{\tau} (\Lambda)^{l} (\rho_{\psi})], \tau \in \{ \alpha, \mu \}$ represents the survival probabilities of input state $\rho_{\psi}$ when the twirled quantum channels $\mathds{E}_{\tau}$ are applied to it $l$ times. 
\end{defn}
Now, we present the following Lemma's and their proofs. 
\begin{lem}
\label{lem:twir-chan}
If $\mathds{E}_{\alpha}(\Lambda)$ is the twirled channel of $\Lambda$ over a set of unitaries $\{U_{\alpha}\}$ spread according to a probability distribution $\alpha$ and $\mathds{E}_{\mu}$ is the Haar-twirl of that channel, then for an $\epsilon$-approximate 2-design, it holds that:
\begin{equation}
\label{eqn:epsilonrho}
    \|\mathds{E}_{\alpha}(\Lambda)(\rho) - \mathds{E}_{\mu}(\Lambda)(\rho)\|_{1} \leq \epsilon \enspace ,
\end{equation}
with $\mathds{E}_{\alpha}(\Lambda) (\rho) = \int_{U} d\alpha(U) U \circ \Lambda \circ U^{\dagger}(\rho)$.
\end{lem}
\begin{proof}[Proof for Lemma~\ref{lem:twir-chan}]
From Definitions~\ref{def:two-design} and \ref{def:diamond-norm} we obtain:
\begin{equation}
    \|\mathds{E}_{\alpha}(\Lambda) - \mathds{E}_{\mu}(\Lambda)\|_{1} \leq \|\mathds{E}_{\alpha}(\Lambda) - \mathds{E}_{\mu}(\Lambda)\|_{\diamond} \leq \epsilon 
    \enspace ,
\end{equation}
which implies:
\begin{equation}
    \|\mathds{E}_{\alpha}(\Lambda) - \mathds{E}_{\mu}(\Lambda)\|_{1} \leq \epsilon
    \enspace . 
\end{equation}
Using Definition~\ref{def:trace-norm} we state:
\begin{equation}
    \|\mathds{E}_{\alpha}(\Lambda) - \mathds{E}_{\mu}(\Lambda)\|_{1} := \max_{\rho} \|\mathds{E}_{\alpha}(\Lambda)(\rho) - \mathds{E}_{\mu}(\Lambda)(\rho)\|_{1}
    \enspace .
\end{equation}
And it holds by definition that:
\begin{equation}
\begin{split}
    &\|\mathds{E}_{\alpha}(\Lambda)(\rho) - \mathds{E}_{\mu}(\Lambda)(\rho)\|_{1} \\
    \leq &\max_{\rho} \|\mathds{E}_{\alpha}(\Lambda)(\rho) - \mathds{E}_{\mu}(\Lambda)(\rho)\|_{1} 
    \enspace .
\end{split}
\end{equation}
Therefore:
\begin{equation}
\centering
    \begin{split}
    &\|\mathds{E}_{\alpha}(\Lambda)(\rho) - \mathds{E}_{\mu}(\Lambda)(\rho)\|_{1} \\
    \leq &\max_{\rho} \|\mathds{E}_{\alpha}(\Lambda)(\rho) - \mathds{E}_{\mu}(\Lambda)(\rho)\|_{1} \\
    \leq &\|\mathds{E}_{\alpha}(\Lambda) - \mathds{E}_{\mu}(\Lambda)\|_{1} 
    \leq \epsilon
        \enspace ,
    \end{split}
\end{equation}
as required. 
\end{proof}
\begin{lem}
\label{lem:surv-prob-ed}
If the unitaries $\{U_{\alpha}\}$ form an $\epsilon$-approximate 2-design, it holds that:
\begin{equation}\label{eq:delta_Pl}
    |P_{l}^{\alpha} - P_{l}^{\mu}|=|\delta P_l| \leq l \cdot \epsilon \enspace .
\end{equation}
\end{lem}
\begin{proof}[Proof for Lemma~\ref{lem:surv-prob-ed}]
Considering a fixed length $l$, we define the state after a sequence of this length, $S_{l}$, of imperfect unitaries $\Lambda_{U_{l}} = \Lambda_{e} \circ U_{l}$, has been applied to initial state $\rho_{\psi} = \ket{\psi}\bra{\psi}$ and before a measurement has been taken, as:
\begin{eqnarray}
\rho(S_{l}) &:=&\Lambda_e \circ {U_1}^{\dagger} \cdots {U_{l}}^{\dagger} \circ \Lambda_e \circ U_l \circ \Lambda_e\circ U_{l-1}\cdots \circ \Lambda_e\circ U_1 (\ket{\psi}\bra{\psi})
\enspace .
\end{eqnarray}
We set $\rho_0$ as the average state of the RB protocol before the final measurement (i.e. averaging the above expression over different sequences according to the distribution $d\alpha$):
\begin{equation}
\rho_0=\int d\alpha(U_1)\cdots d\alpha(U_l)\rho(S_l)   \enspace . \end{equation}
which is equivalent to the average channel defined in Lemma~\ref{lem:twir-chan} but for a sequence of length $l$, i.e. $\mathds{E_{\alpha}}(\Lambda)^{l}(\rho)$. Similarly, we define the following:
\begin{eqnarray}
\rho_1&=&\int d\alpha(U_1)\cdots d\alpha(U_{l-1})d\mu(U_l)\rho(S_l)\nonumber\\
&\cdots&\nonumber\\
\rho_j&=&\int d\alpha(U_1)\cdots d\alpha(U_{j-1})d\mu(U_j) \nonumber \\
& & \cdots d\mu(U_l)\rho(S_l) \nonumber \\
&\cdots& \nonumber\\
\rho_l&=&\int d\mu(U_1)\cdots d\mu(U_l)\rho(S_l)
\enspace .
\end{eqnarray}
Here, in each of the above quantum states we replace (one-by-one) the average over the distribution $\alpha$ with that of the Haar measure $\mu$. Any two consecutive states $\rho_j,\rho_{j+1}$ differ by a single integration, and by the $\epsilon$-approximate 2-design property (see Eq.~\ref{eqn:epsilonrho}) we, therefore, have that:
\begin{equation}
\label{eqn:tracestates}
\begin{split}
&\|\rho_j-\rho_{j+1}\|_{1} \leq \epsilon    \\
\implies 
&\|\rho_l-\rho_{0}\|_{1}\leq l\cdot\epsilon    
\enspace ,
\end{split}
\end{equation}
where the implication follows from the triangle inequality. \\ \\
From the definition of the trace-norm we get: 
\begin{eqnarray}
|\bra{\psi}\rho_l\ket{\psi}-\bra{\psi}\rho_{0}\ket{\psi}|&\leq& l\cdot\epsilon\nonumber\\
|P_l^\mu-P_l^\alpha|&\leq& l\cdot\epsilon
\enspace .
\end{eqnarray}
The definition of survival probability stated in Eq.~\ref{eqn: traceprobability} results in the the equivalence of the above equation. Intuitively, the difference in the probabilities, $\bra{\psi} \rho_{l} \ket{\psi}$ and $\bra{\psi} \rho_{0} \ket{\psi}$, of obtaining the states (after measurement) is no larger than $l \cdot \epsilon$, the bound specified by the trace-norm between the the two initial states, as in Eq.~\ref{eqn:tracestates}.
\end{proof}
Thus completing the proof. 
\subsection{Proof of Theorem\texorpdfstring{~\ref{thm: final_bound}}{III.2}}
\begin{lem} 
\label{lem: delta_f-ed}
Let $A$ and $B$ be known quantities. Under Assumption~\ref{assum: statisticalerror} and with a small $\epsilon$, the error in determining $f_{l}$ from the RB method of an $\epsilon$-approximate 2-design, is given by:
\begin{equation}
    \delta f_{l} \approx\frac{\epsilon}{f^{l-1}B} \enspace ,
\end{equation}
where $f$ is dependent on $l$, $f_{l} = f \pm \delta f_{l}$ is the value for the fidelity decay parameter found for RB with an $\epsilon$-approximate 2-design and $f$ is that obtained with an exact design. 
\end{lem}
Note that in the following proofs we denote the error in obtaining a quantity $C$ as $\delta C$, where a subscript is added if the error depends on some measured quantity that is not obvious. 
\begin{proof}[Proof for Lemma \ref{lem: delta_f-ed}]
From Definition \ref{def:surv-prob-ed} we have that: 
\begin{equation}
P_{l}^{\mu} \pm \delta P_{l} = A + B (f \pm \delta f_{l})^{l} 
\enspace .
\end{equation}
With Lemma \ref{lem:surv-prob-ed} and under Assumption \ref{assum: statisticalerror}, it follows that:  
\begin{eqnarray}
P_l^\mu \pm l\cdot \epsilon &=& A+ B (f\pm\delta f_l)^l\nonumber\\
f^l\pm \frac{l\cdot \epsilon}{B}&=&(f\pm\delta f_l)^l\nonumber\\
f^l\pm \frac{l\cdot \epsilon}{B}&\approx&f^l\pm l\cdot f^{l-1}\delta f_l
\enspace ,
\end{eqnarray}
where the last approximation holds if $\delta f_l/ f \ll 1$. This leads to:
\begin{equation}
\label{eqn: f_l}
\delta f_l=\frac{\epsilon}{f^{l-1}B}    
\enspace .
\end{equation}
\end{proof}
Note that $f_{l}$ and $\delta f_{l}$ can be computed separately for each different length $l$. In practice, $B$ will also depend on the SPAM errors which are (generally) unknown, and therefore it is essential to consider several different lengths to obtain a value for $f_{l}$.
\begin{lem}
\label{lem:delta_r-ed}
The error in determining the average error-rate $r_{l}$ of a gate-set that is an $\epsilon$-approximate 2-design, from the RB method, is given by:
\begin{equation}
    \delta r_{l}=\frac{\epsilon}{p_s f^{l-1}} \enspace ,
\end{equation}
as compared to that when using RB with an exact 2-design. Where $p_s$ is the parameter that characterises the depolarising channel of the SPAM errors (i.e. if no SPAM errors are present, $p_s=1$, while if SPAM errors completely depolarise the channel $p_s=0$).
\end{lem}
\begin{proof}[Proof for Lemma~\ref{lem:delta_r-ed}]
Similarly, given that $r=\frac{d-1}{d}(1-f)$, $B=\frac{(d-1)p_s}{d}$ and the previous result $\delta f_l=\frac{\epsilon}{B f^{l-1}}$, we obtain:
\begin{eqnarray}
\label{eqn: r_l}
r\pm\delta r_l&=&\frac{B}{p_s}(1-(f\pm \delta f_l))\nonumber\\
\delta r_l&=&\frac{\epsilon}{f^{l-1}p_s}
\enspace .
\end{eqnarray}
\end{proof}
It is clear to see that the larger the $l$ that we use to estimate $r_{l}$, the larger the error in that estimation due to the gate-set being an $\epsilon$-approximate 2-design. In our simulated results, we assume no SPAM errors ($p_s=1$), and we can therefore take the more optimistic view and consider the errors for $l=1$. Under these assumptions, it is easily seen that:
\begin{lem}
\label{lem: boundsonr}
If $l, p_{s} = 1$ the value of $r$ determined from the RB method when testing an $\epsilon$-approximate 2-design, is bounded as follows:
\begin{equation}
\begin{split}
    r' &= r \pm \delta r \\ 
    r -\epsilon &\leq r' \leq r +\epsilon \enspace ,\\
    \label{eqn: boundsforr}
\end{split}
\end{equation}
where we use $r'$ to denote the value for the average error-rate measured from using an $\epsilon$-approximate 2-design, and $r$ is the value gained when using an exact 2-design, with the RB method.
\end{lem}
For simplification, we assume no SPAM errors, $p_s=1$, in our analysis; however, it is clear that the error in $f_{l}$ and $r_{l}$ would increase with greater SPAM errors. By extrapolating the values of $f'$ and $r'$ from the survival probability of the smallest length, $l=1$, Eq.~\ref{eqn: r_l} allows us to find the smallest error. In practise the SPAM errors $p_s$ exist and we need multiple values of $l$ to estimate and remove this contributing factor from $p_{s}$. This is not necessary for our purposes and we can take the weakest bound on the average error rate, where $l=1$, and is given by $\delta r_{min}=\epsilon$, while $\delta f_{min}=\epsilon/B$.
\begin{proof}[Proof for Lemma \ref{lem: boundsonr}]
When we set $l = 1$ and $p_{s} = 1$, the error induced in the fidelity decay parameter $f$ and the average error-rate $r$ become: 
\begin{eqnarray}
\delta f &=& \frac{\epsilon}{B} \nonumber \\
\delta r &=& \epsilon \enspace , \\ 
\end{eqnarray}
and therefore: 
\begin{equation}
r - \epsilon \leq r' \leq r + \epsilon
\enspace ,
\end{equation}
where $r' = r \pm \delta r$ is the average error-rate found from RB using an $\epsilon$-approximate 2-design and $r$ is that found from an exact 2-design. 
\end{proof}

\end{document}